\RequirePackage{fix-cm}
\documentclass[a4paper,UKenglish]{lipics}                     
\hypersetup{%
   breaklinks,%
   colorlinks=true,%
   linkcolor=[rgb]{0.45,0.0,0.0},%
   urlcolor=[rgb]{0.05,0.390,0.0},%
   citecolor=[rgb]{0,0,0.45}                                          
}%
                                                                      
\usepackage{microtype}
\usepackage{lmodern} 

\usepackage{amssymb,latexsym,amsthm}
\usepackage{adjustbox}
\usepackage{algorithm}

\usepackage{enumerate}
\usepackage{appendix}
\usepackage{amsmath}
\usepackage{enumitem} 
\usepackage[numbers]{natbib}

\usepackage{etoolbox}
\makeatletter
    \pretocmd{\NAT@citexnum}{\@ifnum{\NAT@ctype>\z@}{\let\NAT@hyper@\relax}{}}{}{}
\makeatother

\newcommand{\CCite}[1]{\citeauthor{#1}~\cite{#1}}

\usepackage{xspace} 
\usepackage{color}%

\usepackage{subfig}
\usepackage{pgf,tikz}

\usepackage{thmtools, thm-restate} 
\usepackage[nameinlink]{cleveref} 

\usepackage[noend]{algorithmic}
\Crefname{ALC@unique}{Line}{Lines} 

\crefname{ineq}{inequality}{inequalities}
\creflabelformat{ineq}{#2{\upshape(#1)}#3} 
\crefname{claim}{claim}{claims}           
\crefname{defn}{definition}{definitions}  
\crefname{lemma}{lemma}{lemmas}

\newlist{clenum}{enumerate}{1} 
\setlist[clenum]{label=(\alph*),ref=\textup{\theclaim~(\alph*)}}
\crefalias{clenumi}{claim} 
\newlist{thenum}{enumerate}{1} 
\setlist[thenum]{label=(\alph*),ref=\textup{\thetheorem~(\alph*)}}
\crefalias{thenumi}{theorem} 

\usetikzlibrary{arrows,shapes,calc,intersections,through,backgrounds,patterns}
\newtheorem{lem}{Lemma}
\newtheorem{claim}{Claim}[section]
\newtheorem{defn}{Definition}[section]

\newcommand{\Real}{\ensuremath{\mathbb{R}}\xspace}%
\newcommand{\Realp}{\ensuremath{\mathbb{R}^{+}}\xspace}%

\newcommand{\mmc}{\text{MetricMultiCover}\allowbreak\ensuremath{(X, Y, k, \alpha)}}
\newcommand{\spec}{\text{ComputeServerSubsets}\allowbreak\ensuremath{(X, Y, k)}}
\newcommand{\specn}{\text{ComputeServerSubsets}\allowbreak\ensuremath{(X, Y, \kappa)}}
\newcommand{\mmcn}{\text{NonUniformCover}\allowbreak\ensuremath{(X, Y, \kappa, \alpha)}}
\newcommand{\lgt}{\text{largest}}

\newcommand{\yd}[2]{\ensuremath{y_{#1}({#2})}}

\newcommand{\oX}{\ensuremath{\overline{X}}}

\newcommand{\cost}{\texttt{cost}}
\newcommand{\thr}{\ensuremath{\texttt{th}}}

\newcommand{\rhi}[1]{\ensuremath{\rho_{i}(#1)}}
\newcommand{\RHI}[1]{\ensuremath{\rho_{\lambda_{#1}}}}

\newcommand{\YN}[2]{\ensuremath{N_{#2}(#1)}}

\newcommand{\GL}[1]{\ensuremath{G_{\lambda_{#1}}}}
\newcommand{\HL}[1]{\ensuremath{H_{\lambda_{#1}}}}
\newcommand{\XL}[1]{\ensuremath{X_{\lambda_{#1}}}}
\newcommand{\A}[2]{\ensuremath{A_{#2}{(#1)}}}
\newcommand{\K}[1]{\ensuremath{\lambda_{#1}}}
\newcommand{\F}{\ensuremath{{\cal{F}}}\xspace}
\newcommand{\Ys}[1]{\ensuremath{Y^s_{\lambda_{#1}}}}
\newcommand{\Yp}[1]{\ensuremath{Y^p_{\lambda_{#1}}}}
\newcommand{\Cover}{\text{Cover}}
\newcommand{\shat}[1]{\ensuremath{{#1}_1}}
\newcommand{\ns}[2]{\texttt{nn}({#1}, {#2})}
\newcommand{\NN}[2]{\texttt{NN}({#1}, {#2})}
\newcommand{\MU}[1]{\ensuremath{\mu_{\lambda_{#1}}}}
\newcommand{\DOWNTO}{\textbf{downto }}

\def\disk{\delta}

\begin{document}
\title{On Metric Multi-Covering Problems}

\author{Santanu Bhowmick\thanks{santanu-bhowmick@uiowa.edu}}
\author{Tanmay Inamdar\thanks{tanmay-inamdar@uiowa.edu}}
\author{Kasturi Varadarajan\thanks{kasturi-varadarajan@uiowa.edu}}
\affil{
  Department of Computer Science\\
  University of Iowa, Iowa City, USA\\
  }
  
\authorrunning{S.\,Bhowmick, T.\,Inamdar and K.\,Varadarajan} 

\Copyright{Santanu Bhowmick, Tanmay Inamdar and Kasturi Varadarajan}

  
\maketitle
\begin{abstract}
   In the metric multi-cover problem (MMC), we are given two
   point sets $Y$ (servers) and $X$ (clients) in an arbitrary metric space $(X \cup Y, d)$, a
   positive integer $k$ that represents the coverage demand of each client, and
   a constant $\alpha \geq 1$. Each server can have a single ball of arbitrary
   radius centered on it. Each client $x \in X$ needs to be covered by at least
   $k$ such balls centered on servers. The objective function that we wish to
   minimize is the sum of the $\alpha$-th powers of the radii of the balls. 

   In this article, we consider the MMC problem as well as some non-trivial generalizations, such
   as (a) the {\em non-uniform MMC}, where we allow client-specific demands, and (b) the {\em $t$-MMC}, where we require the number of 
   open servers to be at most some given integer $t$. For each of these
   problems, we present an efficient algorithm that reduces the problem to
   several instances of the corresponding $1$-covering problem, where the
   coverage demand of each client is $1$.  Our reductions preserve optimality
   up to a multiplicative constant factor.

   Applying known constant factor approximation algorithms for $1$-covering, we obtain the first constant approximations for the MMC and these generalizations.
\end{abstract}
\section{Introduction}
\label{sec:intro}
In the metric multi-cover problem (MMC), the input consists of two point sets $Y$ (servers) and $X$ (clients) in an arbitrary metric
space $(X \cup Y, d)$, a positive integer $k$ that represents the coverage
demand of each client, and a constant $\alpha \geq 1$. 

Consider an assignment $r: Y \rightarrow \Real^+$ of {\em radii} to each server
in $Y$. This can be viewed as specifying a ball of radius $r(y)$ at each server $y \in Y$. If, for each client $x \in X$, at least $j$ of the corresponding server balls contain $x$, then we say that $X$ is $j$-covered by $r$. That is, $X$ is $j$-covered if for each $x \in X$, 
\[ 
    | \{ y \in Y \ | \ d(x, y) \leq r(y) \} | \geq j.
\]
The {\em cost} of assignment $r$ is defined to be the sum of the $\alpha$-th powers of radii of the corresponding balls. That is, $\cost(r) = \sum_{y \in Y} \left ( r(y)
\right )^{\alpha}$.

Any assignment  $r: Y \rightarrow \Real^+$ that $k$-covers $X$ is a feasible
solution to the MMC problem. The objective to be minimized is the cost  $\sum_{y
\in Y} \left ( r(y) \right )^{\alpha}$.  We assume that $k \leq |Y|$, for
otherwise there is no feasible solution. 

In this article, we consider the MMC as well as some more general variants. In
one variant, the {\em non-uniform MMC}, we allow each client to specify its own
coverage requirement. In another, called the {\em $t$-MMC problem}, we require
that the number of servers used is at most some input integer $t$. Here, the
algorithm will have to determine which subset $Y' \subseteq Y$ containing at
most $t$ servers to use, and then $k$-cover $X$ using that subset. These
problems are \textsf{NP}-hard in general, as we point out below, and we are
interested in approximation algorithms that run in polynomial time.  

\subsection{Prior Work}
\label{ssec:related}
We review some of the previous works which considered the MMC problem and its
variants. A version of the MMC problem arises naturally in fault-tolerant
wireless sensor network design. The clients ($X$) and the servers ($Y$) are
points in the plane, and the requirement is to construct disks centered at the
servers such that each client is covered by at least $k$ distinct server disks.
A server disk corresponds to the area covered by some wireless antenna placed at
the server, whose power consumption is proportional to the area being serviced
by the antenna. The objective is to minimize the power consumption of
the antennas placed at the servers while meeting the coverage requirement of
each client. This problem is thus a special case of the MMC problem, with $X,
Y$ in the plane and $\alpha = 2$. This special case has been studied in several
recent works~\cite{Abu-AffashCKM11,BhowmickVX13,BhowmickVX15}. 

~\CCite{Abu-AffashCKM11} considered the special case $\alpha = 2$ of the MMC problem where
$X$ and $Y$ are subsets of $\Real^2$ and the metric is the Euclidean distance.
They gave an $O(k)$ approximation for the problem. (Throughout the paper, it is
implicit that we only refer to polynomial time algorithms.) Following their
work, ~\CCite{BhowmickVX13} gave an $O(1)$ approximation, thus obtaining a
guarantee that is independent of the coverage demand. Their approximation
guarantee also holds for a {\em non-uniform} generalization of the MMC problem
where each client can have an arbitrary demand. The algorithm
in~\cite{BhowmickVX13} was further generalized in the journal article~\cite{BhowmickVX15}, in which
an approximation guarantee of $4 \cdot (27\sqrt{2})^{\alpha}$ was achieved for
the MMC problem in the plane, for any $\alpha \geq 1$. 

In $\Real^d$, their approximation guarantee is $(2d) \cdot (27
\sqrt{d})^{\alpha}$, which depends on the dimension. Motivated by this,
\CCite{BhowmickVX15} ask whether an $O(1)$ guarantee is possible in an arbitrary metric
space. This is one of the questions considered in this article. The metric space setting generalizes not only the Euclidean distance in any dimension, but also the
shortest path distance amidst polygonal obstacles in $\Real^2$ or $\Real^3$, and
in graphs.

~\CCite{Bar-YehudaR13} were the first to give an approximation algorithm for the
MMC problem in which $X,Y$ are points in an arbitrary metric space.
They presented a $3^{\alpha}\cdot k$ approximation guarantee, using the
local-ratio technique. They also consider the non-uniform version of the
problem, where the coverage demand of each client is an arbitrary integer that is not necessarily related to the demands of the other clients. They obtain a $3^{\alpha}\cdot k_{\texttt{max}}$ approximation for this version, where $k_{\texttt{max}}$ is now the maximum client demand. Their guarantees also hold for minimizing a more general objective function  $\sum_{y
\in Y} \left ( w_y r(y) \right )^{\alpha}$, where weight $w_y \geq 0$ is specified for each server $y$ as part of the input; we do not address this objective function here.

The case $k = 1$ for the MMC problem is a traditional covering problem, and has
a much longer history. For the general metric setting (with $k =
1$), the primal-dual method has been successful in obtaining constant factor
approximations, as demonstrated in~\cite{CharikarP04,FreundR03}. In $\Real^d$ for any fixed $d$, a polynomial-time
approximation scheme (PTAS) using dynamic programming exists, as shown
by~\CCite{BiloCKK05} (improving on the work of~\CCite{Lev-TovP05} who obtained a
PTAS for the plane and $\alpha = 1$). The case $\alpha = 1$ is somewhat special -- as shown by \CCite{gibson2012clustering}, the problem can be solved exactly in polynomial time if the underlying metric is $\ell_{\infty}$ or $\ell_1$.  

The MMC problem is known to be \textsf{NP}-hard even when $X, Y$ are point sets in the
plane and $k = 1$, for any $\alpha > 1$. This was established
by~\CCite{BiloCKK05} for $\alpha \geq 2$, and subsequently for $\alpha > 1$
by~\CCite{AltABEFKLMW06}.

\noindent {\bf MMC with cardinality constraints.} To our knowledge, the $t$-MMC  problem, where we are given a bound $t$ on the number of servers that can be
opened, has not been studied in its generality. The special case of $1$-covering ($k = 1$) has,
however, received considerable attention. Here, one wants to find $t$ server
balls to cover the clients, and minimize the sum of (the $\alpha$-th powers of) the radii of the balls; this
may be compared to the $t$-center problem, where one wants instead to minimize
the maximum radius. For this special case of $t$-MMC, Charikar and Panigrahy
\cite{CharikarP04} address the metric setting and give an $O(1)$ approximation.
Although they explicitly address only the case $\alpha = 1$, their guarantee
generalizes to any $\alpha \geq 1$. Exploiting the special structure for the case
$\alpha = 1$,~\CCite{gibson2012clustering} give a polynomial time algorithm for solving the problem exactly in $\Real^d$ if the underlying metric is $\ell_{\infty}$ or $\ell_1$; for the $\ell_2$ metric they obtain a polynomial time approximation scheme.     

\noindent {\bf Related Results.} Fault tolerant versions of other related problems have also been studied in the literature -- facility location \cite{guha2003constant, swamy2008fault, byrka2010fault}, $t$-median \cite{hajiaghayi2016constant}, and $t$-center \cite{khuller2000fault}. Constant factor approximations are known for all these problems in the metric setting. In particular, the results for facility location solve the natural LP-relaxation and perform a clever rounding. These rounding methods do not readily extend to our setting. One reason for this is the fact that we are dealing with a covering problem; another reason is the additional constraint that one has to write in the LP saying that each server can house at most one ball.
  
Some recent results that involve geometric set multi-covering problems can be found in~\cite{ChekuriCH12, BansalP12}. Reducing to a set multi-covering problem does not seem to be an effective way to deal with the MMC, partly because in a feasible solution to the MMC each server can contribute only one ball. This issue is discussed in greater detail in~\cite{BhowmickVX15}.

\subsection{Our Results and Techniques}
We present a polynomial-time algorithm that reduces the MMC to several instances
of the $1$-covering version, where $k = 1$. This reduction preserves optimality
to within a constant multiplicative factor. More specifically, our reduction
outputs {\em pairwise disjoint} subsets $Y_1, Y_2, \ldots, Y_k$ of servers such
that computing an optimal $1$-cover of the clients $X$ using each $Y_i$ and
combining the $1$-covers results in a solution whose cost is $O(1)$ of the optimal cost. 

Using a known constant factor approximation algorithm for computing a $1$-cover,
we obtain an $O(1)$ approximation for the MMC problem in any metric space,
achieving a guarantee that is independent of the coverage demand $k$. This
resolves a problem left open by Bhowmick et al.~\cite{BhowmickVX15}, whose
approximation guarantee in the Euclidean setting depends on the dimension.
Concretely, our approximation guarantee is $2 \cdot (108)^{\alpha}$. We have not attempted to optimize the constants, as our focus is on answering the question of whether a guarantee independent of $k$ and the dimension is possible.

Using the same paradigm of reducing to several $1$-covering instances, we obtain the first $O(1)$ approximation for the non-uniform MMC in the metric setting, as well as the first $O(1)$ approximation for the $t$-MMC. 

We now explain some key ideas in this paper. For a client $x$, and any $1 \leq i
\leq |Y|$, let us define the $i$-neighborhood of $x$, $\YN{x}{i}$, to be the set
consisting of the $i$ nearest servers of $x$. At the core of our reduction is an
analysis of the neighborhoods of the clients that may be of independent
interest. In order to motivate this analysis, we first need to explain our high
level plan for the server subsets $Y_1, Y_2, \ldots, Y_k$ in the MMC. As
observed in \cite{BhowmickVX15}, the optimal MMC solution can be viewed, up to a
constant factor approximation, as a sequence $\rho_{1}, \rho_{2}, \ldots, \rho_{k}$, where each $\rho_{i}$ is a cover of $X$. In particular, $\rho_{i}$ is a special type of cover, called an outer cover of level $i$. This means that for each client $x$, there is a large ball in 
$\rho_{i}$ that contains $x$ -- a ball whose radius is at least as large as 
the distance from $x$ to its $i$-th nearest server.  
 
Our plan for the server subsets $Y_1, Y_2, \ldots, Y_k$ is that for each $1 \leq
i \leq k$, $Y_i$ shall be ``almost'' a hitting set for $\rho_i$. If this can be achieved, then we can obtain a cover of $X$ using just the servers in $Y_i$ by moving each ball in $\rho_{i}$ to a server in $Y_i$ that hits it, and expanding the ball slightly. The cost of this cover is within a constant of that of $\rho_{i}$. Doing this for each $1 \leq i \leq k$, we get $k$ covers of $X$ whose total cost is within a constant of the optimal MMC solution. Furthermore, the fact that the subsets $Y_1, Y_2, \ldots, Y_k$ are pairwise disjoint implies that these $k$ covers together form a valid MMC solution.   

Thus, we would like each $Y_i$ to be a hitting set for the corresponding outer cover $\rho_{i}$. Note, however, that we do not know anything about $\rho_{i}$, as it comes from the unknown MMC optimum. Therefore, we aim for an equivalent goal -- we would like $Y_i$ to be a hitting set for the $i$-neighborhoods of the clients. More concretely, we ask: can we extract $k$ pairwise disjoint server subsets $Y_1, Y_2, \ldots, Y_k$ such that for each $1 \leq i \leq k$ and each client $x$, $\YN{x}{i} \cap Y_i \neq \varnothing$?

This specification is too stringent, and the answer to this question is ``no''. Let $k = 2$, and suppose there are two clients at distance $1$ from each other, one server that is co-located with the first client, and a second server that is co-located with the second client. In this example, both servers would have to be in $Y_1$, leaving no server for $Y_2$.

Thus, we need a weaker specification for the $Y_i$ that is still sufficient for our purposes. To describe it, we need one more notion.  Let $G_i= (X, E_i)$ be the intersection graph of $i$-neighborhoods of $X$ i.e.\
$(x_1, x_2) \in E_i$ iff $\YN{x_1}{i} \cap \YN{x_2}{i} \neq \varnothing$. What we are able to show is the following.

\begin{lem}
\label{lem:spec}
Assume $k$ is even. We can efficiently compute a set
\[ \left(\bigcup_{i = \frac{k}{2} + 1}^k Y_i^s\right) \cup \left(\bigcup_{i = \frac{k}{2} + 1}^k Y_i^p\right)  \]
of $k$ pairwise disjoint server subsets such that for each $\frac{k}{2} + 1 \leq i \leq k$ and each client $x \in X$, there is a client $x'$ within two hops of $x$ in $G_{i}$ such that $\YN{x'}{i} \cap Y_i^{s} \neq \varnothing$ (resp. $\YN{x'}{i} \cap Y_i^{p} \neq \varnothing$).
\end{lem}  

Note that we have weakened the original specification in two ways. First, instead of considering $i$-neighborhoods for each $1 \leq i \leq k$, we only consider
$i$-neighborhoods for each $\frac{k}{2} + 1 \leq i \leq k$, but now require two
hitting sets for each such $i$. Second, for a fixed $\frac{k}{2} + 1 \leq i \leq
k$, we do not require that $Y_i^{s}$ hits the $i$-neighborhood of \emph{every} client in $X$. We only require that for any client $x$, there is some client $x'$ that is `near' $x$ such that $Y_i^{s}$ intersects the $i$-neighborhood of $x'$. The requirement for $Y_i^p$ is also relaxed in this way. The notion of `near' is a natural one -- that of being within a distance of $2$ in the intersection graph $G_i$ of the $i$-neighborhoods.

The proof of~\Cref{lem:spec}, which is given in~\Cref{sec:spec}, is delicate. We
construct the family $Y_k^{s}, Y_k^{p}, Y_{k-1}^{s}, Y_{k-1}^{p}, \ldots, Y_{\frac{k}{2} +
1}^{s}, Y_{\frac{k}{2} + 1}^{p}$ in that order, but we have to be careful while
picking the earlier subsets to ensure that there are suitable servers left for
building the later subsets.

The algorithm for the MMC, which builds on \Cref{lem:spec} as outlined above,
is given in \Cref{sec:algo}. In the case of the non-uniform MMC, the situation
is complicated by the fact that clients can have different demands. Nevertheless, we are able to extend the scheme of extracting disjoint server subsets and reducing to suitable $1$-covering instances.  Because of the varying coverage demands, a generated $1$-covering instance may only involve a subset of the clients. The algorithm for the non-uniform MMC is described in \Cref{sec:conc}.

For the $t$-MMC, which is addressed in \Cref{sec:tMMC}, the computation of
disjoint server subsets is identical to that of the MMC. However, the reduction
to $1$-covering is subtler as we have to worry about how many open servers are
allowed for each $1$-covering instance. One tool we develop to address this
issue is the extraction of $k$ outer covers from the optimal solution with
additional guarantees on the number of servers opened in each outer cover. It is
worth pointing out that in the case of $1$-covering in the context of $t$-MMC,
the only known approximation is the somewhat involved algorithm of Charikar and
Panigrahy \cite{CharikarP04}. Thus, it is especially fortuitous that we are able
to deal with the $t$-MMC by reducing to the $1$-covering case, for which we can
use their algorithm as a black box.

\section{Preliminaries}
\label{sec:terminology}
In this section, we define some notation and some needed tools from prior work.

Let $\disk(p,r)$ denote the ball of radius $r$ centered at $p$, i.e., $\disk(p,
r) = \{ u \in X \cup Y \mid d(p, u) \leq r \}$. For brevity, we slightly abuse
the notation and write $\disk(p, d(p,q))$ as $\disk(p, q)$. The {\em cost} of a
set $B$ of balls, denoted $\cost(B)$, is defined to be the sum of the
$\alpha$-th powers of the radii of the balls. 

Any assignment $r: Y \rightarrow \Real^+$ corresponds to the set of balls $\{
\disk(y, r(y)) \mid y \in Y \}$. Note that the cost of assignment $r$ is the same
as the cost of the corresponding set of balls. Instead of saying that $r$
$j$-covers $X$, we will often say that the corresponding set of balls $j$-covers
$X$. We will say that a set of balls \emph{covers} $X$ instead of saying it $1$-covers $X$. 

For each
$x \in X$ and $1 \leq j \leq |Y|$, we define $\yd{j}{x}$ to be the $j$-th
closest point in $Y$ to $x$ using distance $d$. The ties are broken arbitrarily. For any $x \in X$, we define the $i$-neighborhood ball of a
client $x$ as $\disk(x, \yd{i}{x})$. We define the
$i$-neighborhood of $x$, $\YN{x}{i}$,  as $\{ \yd{j}{x} \mid 1 \leq j \leq i
\}$. 

For $1 \leq i \leq k$, let $G_i= (X, E_i)$ be the intersection graph of $i$-neighborhoods of $X$ i.e.\
$(x, x') \in E_i$ iff $\YN{x}{i} \cap \YN{x'}{i} \neq \varnothing$.

\subsection{Computing \texorpdfstring{$1$}{1}-covers}
We will need as a black-box an algorithm that, given subsets $X' \subseteq X$ and $Y' \subseteq Y$, computes a $1$-cover of $X'$ using servers in $Y'$. That is, the algorithm must return an assignment of
radii $r: Y' \rightarrow \Realp$ such that each client $x \in X'$ is contained in at least one ball centered on a server in $Y'$. Computing a  $1$-cover of minimum cost is thus the
special case for the MMC problem where $k = 1$. As mentioned
in~\Cref{ssec:related}, even this version is \textsf{NP}-hard, but it does admit
constant-factor approximations~\cite{CharikarP04,Bar-YehudaR13}. Let
$\text{Cover}(X', Y', \alpha)$ denote an algorithm that returns a $1$-cover of 
$X'$ using servers in $Y'$ with cost at most
$3^{\alpha}$ times the cost of an optimal $1$-cover.  

\subsection{Outer Cover}
Our work also relies on the notion of an \textit{outer cover}, which is
described in~\CCite{BhowmickVX15}. We adopt the definition of an outer cover
from~\cite{BhowmickVX15} as follows:
\begin{defn}
\label{def:OC}
Given point sets $X, Y$ in a metric space $(X \cup Y, d)$, positive integer $i$
and $\alpha \geq 1$, an outer cover of level $i$ is an assignment $\rho_i: Y
\rightarrow \Real^+$ of radii to the servers such that for each client $x \in
X$, there is a server $y \in Y$ such that
\begin{enumerate}
        \item The ball $\delta(y, \rhi{y})$ contains $x$ i.e.\ $d(y, x) \leq
            \rhi{y}$.
        \item Radius of the ball at $y$ is large, that is, $\rhi{y} \geq d(x, \yd{i}{x})$
\end{enumerate}
\end{defn}

Given a level $i$ outer cover $\rho_i$, and a client $x \in X$, any server $y$ that 
satisfies the two conditions in the definition above is said to {\em serve} $x$;
we also say that the corresponding ball $\delta(y, \rhi{y})$ serves $x$.

To appreciate why outer covers play an important role, consider any $k$-cover of
the set $X$ of clients. Fix $1 \leq i \leq k$. Form a set $B$ of balls by
adding, for each client in $X$, the $i$-th largest ball in the $k$-cover that
covers the client. The set $B$ thus constructed is seen to be a level $i$ outer
cover.

The sum of the costs of the optimal $i$-th level outer covers, for $1 \leq i \leq k$,
gives a lower bound on the cost of the optimal solution to the MMC. This is
stated precisely in the theorem below. The proof can be found in
\cite{BhowmickVX15}, but for completeness, it is extracted and given
in~\Cref{sec:outer-cover}.

\begin{restatable}{theorem}{outercover}
    \label{lem:lbound}
    Let $r': Y \rightarrow \Real^+$ be any assignment that constitutes a
    feasible solution to the MMC problem. For each 
    $1 \leq i \leq k$, let $\mu_i$ denote the cost of an optimal outer cover of level $i$. Then
 \[ \sum\limits_{i = 1}^k \mu_i \leq 3^{\alpha} \cdot
            \cost(r').\]

\end{restatable}

\section{Partitioning Servers}
\label{sec:spec}
Suppose that we are given two point sets $Y$ (servers) and
$X$(clients) in an arbitrary metric space $(X \cup Y, d)$, and a positive 
integer $k$ that represents the coverage demand of each client, and the constant $\alpha \geq 1$. In this section, we establish the following result, which is \Cref{lem:spec} restated so as to also address the case where $k$ is odd.

\begin{lem}
\label{lem:spec-full}
Let $l = \lceil k/2 \rceil$. We can efficiently compute a family $\F$ of $k$ server subsets such that 
\begin{enumerate}
\item $\F$ contains two subsets $Y_{i}^{s}$ and $Y_{i}^{p}$ for each $l + 1 \leq i \leq k$, and, if $k$ is odd, one additional subset $Y_{l}^{p}$.
\item $\F$ is a pairwise disjoint family, i.e., any two subsets in $\F$ are disjoint.
\item Suppose that (a) $l + 1 \leq i \leq k$ and $Y_i$ is either $Y_{i}^{s}$ or
$Y_{i}^{p}$, or (b) $k$ is odd, $i = l$, and $Y_i = Y_{i}^{p}$. For any client $x \in X$, there is a client $x'$ within two hops of $x$ in $G_{i}$ such that $\YN{x'}{i} \cap Y_i \neq \varnothing$.
\end{enumerate}
\end{lem}

Before describing the algorithm for computing the family $\F$, we introduce needed concepts. For a positive integer $r$, an $r$-net of a graph $G = (V, E)$ is a set $S \subseteq V$
such that every path in $G$ between any two vertices in $S$ has at least
$r$ edges in it, and for every $u \in V \setminus S$, there exists a vertex
$v \in S$ such that $u$ is reachable from $v$ using a path in $G$ having at most
$r - 1$ edges. An $r$-net is a fairly well-known concept; for instance, a $2$-net is simply an independent set that is maximal by inclusion.

We note that for any $G_i$, $G_j$ such that $l \leq i < j \leq k$, $G_i$
is a sub-graph of $G_j$ since the $i$-neighborhood of any client is 
contained within its $j$-neighborhood. Motivated by the statement of \Cref{lem:spec-full}, we would like to compute a $3$-net $X_i$ of $G_i$, for each $1 \leq i \leq k$. This would ensure that for any client $x \in X$, there is a client $x' \in X_i$ that is within two hops of $x$ in $G_i$.
For the rest of this section, we refer to a $3$-net as simply a net.

\begin{claim}
\label{cl:ruling}
There is a polynomial time algorithm that, given $X$, $Y$, and $k$, computes a
hierarchy
\[
    X_k \subseteq X_{k-1} \subseteq \dots \subseteq X_{2} \subseteq X_1,
\]
where each $X_i \subseteq X$ is a  $3$-net of $G_i$. 
\end{claim}

\begin{proof}
Given a net $X_i$ of $G_i$, we describe how to compute a net $X_{i-1}$ of
$G_{i-1}$ such that $X_i \subseteq X_{i-1}$. Since $G_{i-1}$ is a subgraph of
$G_i$, we have that the (hop) distance in $G_{i-1}$ between any two vertices  in
$X_i$ is at least $3$. We initialize $X_{i-1}$ with $X_i$ and assume that all
vertices in $G_{i-1}$ are initially unmarked. We repeat the following process
till $G_{i-1}$ does not contain any unmarked vertices: mark all vertices in
$G_{i-1}$ within distance $2$ of $X_{i-1}$, and then add an arbitrary unmarked
vertex from $G_{i-1}$ to $X_{i-1}$. 

We can construct the hierarchy of nets by starting with an arbitrary net $X_k$ of the graph $G_k$, and then constructing the successive nets in the hierarchy by the process described above. To construct $X_k$ itself, we apply the above method after initializing $X_k$ to be the singleton set consisting of any vertex in $G_k$. 
\end{proof}

\subsection{Computing Disjoint Server Subsets}
\label{ssec:algo}
   
Our algorithm for computing the family $\F$ of server subsets, as claimed in \Cref{lem:spec-full}, is described in~\Cref{alg:HC}. We begin by
setting parameter $l$ to be $\lceil k/2 \rceil$, just as in the statement of \Cref{lem:spec-full}. We then use \Cref{cl:ruling}
to compute a hierarchy of nets, truncating it at $l$: $X_k \subseteq X_{k-1} \subseteq \cdots
\subseteq X_l$. Any client that belongs to
$\bigcup\limits_{i = l}^k X_i$ is termed as a {\em net client}. For each
client $x$, we denote the $l$-neighborhood $\YN{x}{l}$ as the {\em
private servers} of $x$.

\begin{algorithm*}[hbt]
 \caption{\spec}
 \label{alg:HC}
 \begin{algorithmic}[1]
 \setcounter{ALC@unique}{0}
    \STATE For each $y \in Y$, mark $y$ as available.
    \STATE $l \leftarrow \lceil k / 2 \rceil$ \label{line:l}
    \STATE Compute $X_k \subseteq X_{k-1} \subseteq \cdots \subseteq X_l$ using Claim~\ref{cl:ruling}.
    \FOR {$i = k$ \DOWNTO $l$} \label{lin:outer-for-start}
        \STATE Let $Y^s_i \leftarrow \varnothing, Y^p_i \leftarrow \varnothing$.
        \FORALL {$x_c \in X_i$} \label{lin:inner-for-start}
           \IF {$i > l$}
              \STATE $y_s \leftarrow \mbox{ farthest available server in }
                      \YN{x_c}{i}$. \label{lin:shared-server}
              \STATE $Y^s_i \leftarrow Y^s_i \cup \{y_s\}$. Mark $y_s$ as not
                 available.
           \ENDIF
           \IF {$i > l$ or ($i = l$ and $k$ is odd)}
              \STATE $y_p \leftarrow \mbox{ any available server in }
                     \YN{x_c}{l}$. \label{lin:priv-server}
              \STATE $Y^p_i \leftarrow Y^p_i \cup \{y_p\}$. Mark $y_p$ as not
                 available. \label{lin:phase3-end}
           \ENDIF 
        \ENDFOR 
    \ENDFOR

    \STATE $\F \leftarrow \varnothing$.
    \FOR {$i = k$ \DOWNTO $l + 1$} 
       \STATE $\F \leftarrow \F \cup \{ Y^s_i, Y^p_i \}$.
    \ENDFOR

    \IF {$k$ is odd}
       \STATE $\F \leftarrow \F \cup \{ Y^p_l \}$.
    \ENDIF
    \RETURN The family $\F$
 \end{algorithmic}
\end{algorithm*} 

The disjoint server subsets are computed in 
\Crefrange{lin:outer-for-start}{lin:phase3-end} of~\Cref{alg:HC} -- the for
loop, whose index $i$ goes down from $k$ to $l$. In each iteration $i \geq l + 1$, we extract two disjoint sets of servers $Y^p_i$ and $Y^s_i$, and if $k$ is odd, we extract one server set $Y^p_l$ in iteration $l$.  Notice that when 
summed over all $i$ from $k$ to $l$, we get $k$ disjoint server sets. The algorithm then adds all these server subsets to 
$\F$ and returns it. 

Observe that in iteration $i$ of \Cref{lin:outer-for-start}, we go through each client in $x_c \in X_i$, and use a carefully designed rule to pick two available servers from the $i$-neighborhood $\YN{x_c}{i}$ of $x_c$ to add to $Y^p_i$ and $Y^s_i$. Observe that we add the farthest available server from the $i$-neighborhood $\YN{x_c}{i}$ to $Y^s_i$, whereas we pick an available server from 
$\YN{x_c}{l} \subseteq \YN{x_c}{i}$, i.e., a private server of $x_c$, to add to $Y^p_i$. These choices -- farthest and private -- are crucial to our algorithm. The two  added servers are immediately made unavailable. The fact that $X_i$ is a net of $G_i$
is useful in controlling the impact on server availability for later iterations
of the algorithm. The subsequent section is devoted to establishing the crucial
fact that such available servers can be found in iteration $i$.

Assuming that servers are available whenever the algorithm looks for them, we can now establish \Cref{lem:spec-full}. Fix an $i$ such that $l + 1 \leq i \leq k$, and consider any client $x \in X$. Since $X_i$ is a net of $G_i$, there is a client $x' \in X_i$ that is within two hops of $x$ in $G_i$.  From the inner loop (\Cref{lin:inner-for-start}) in iteration $i$ of the outer loop (\Cref{lin:outer-for-start}), it is evident that for each $x_c \in X_i$, there is (at least) one server in $Y^p_i$ (resp. $Y^s_i$) that belongs to the $i$-neighborhood $\YN{x_c}{i}$. In particular, $\YN{x'}{i} \cap Y_i^p \neq \varnothing$, and
$\YN{x'}{i} \cap Y_i^s \neq \varnothing$. If $k$ is odd, a similar argument 
can be made for $i = l$ and $Y_l^p$. This establishes \Cref{lem:spec-full}, assuming server availability.

\subsection{Server Availability}
\label{sec:sa}
Fix an iteration $i$ of the for loop in \Cref{lin:outer-for-start} in \Cref{alg:HC}. In such an iteration, the algorithm considers each $x_c \in X_i$ in the inner for loop in \Cref{lin:inner-for-start}. For each $x_c$, it looks for up 
to two available servers within $\YN{x_c}{i}$ and uses them. In order
for the algorithm to be correct, such available servers must exist when the
algorithm looks for them. In this section, which is the core of our analysis, we show that this is indeed the case.

Let us begin with a roadmap of this argument. Consider a client $x \in X_k$ that
belongs to the net for $G_k$. Since the nets form a hierarchy, the client $x$
also belongs to the net $X_i$ for each $i < k$. Since the $i$-neighborhoods of
clients in $X_i$ are disjoint, for each $i$, the server choices made by other
net clients do not affect $x$ at all, and so $x$ will be able to find available
servers within $\YN{x}{i}$ for each $i$. Now consider a server $x'$ that first
appears in the net $X_j$ for some $j < k$. That is, $x'$ is not in $X_i$ for any
$i > j$ but is in $X_i$ for every $i \leq j$. What we argue is that at the
beginning of iteration $j$ of the for loop in \Cref{lin:outer-for-start}, the
$j$-neighborhood of $x'$ is, from the perspective of available servers, similar
to that of the $j$-neighborhood of $x$. It is in this argument that we use the
fact that a private server is chosen in \Cref{lin:priv-server}.

\vspace{0.1in}

\noindent {\bf Properties of Nets.} We now state some straightforward 
properties concerning the hierarchy of nets $X_k \subseteq X_{k-1} \subseteq \cdots
\subseteq X_l$. 

     
\begin{claim}
   Let $x, x'$ be two distinct clients in $X_i$. Then $\YN{x}{i}
   \cap \YN{x'}{i} = \varnothing$.
   \label{cl:chosen-chosen}
\end{claim}
\begin{proof}
   Since $X_i$ is a $3$-net of $G_i$, any path between $x$ and $x'$ in $G_i$
   has at least three edges. Recall that the condition $\YN{x}{i} \cap
   \YN{x'}{i} \neq \varnothing$ is equivalent to $(x,x')$ being an edge in
   $G_i$. 
\end{proof} 

\begin{claim}
   Let $x \in X \setminus X_i$. Then there is at most one $x' \in X_i$ such
   that $\YN{x}{i} \cap \YN{x'}{i} \neq \varnothing$.
   \label{cl:chosen-unchosen}
\end{claim}
\begin{proof}
   If there are two clients $x_1$ and $x_2$ in $X_i$ such that $\YN{x}{i} \cap
   \YN{x_1}{i} \neq \varnothing$ and $\YN{x}{i} \cap \YN{x_2}{i} \neq
   \varnothing$, then there is a path in $G_i$ with at most two edges connecting
   $x_1$ and $x_2$. Since the clients $x_1$ and $x_2$ belong to $X_i$, this
   would contradict the fact that $X_i$ is a $3$-net.
\end{proof} 
  
\begin{claim}
\label{cl:subset}
   Let $x_i \in X_i$ and $x_j \in X_j$ be any two distinct clients for $l
   \leq i < j \leq k$. Then, $\YN{x_i}{i} \cap \YN{x_j}{l} = \varnothing$.
\end{claim}
\begin{proof}
   Since $i < j$, we have $X_j \subseteq X_i$ and hence the clients $x_i$ and
   $x_j$ both belong to the net $X_i$, implying that $\YN{x_i}{i} \cap
   \YN{x_j}{i} = \varnothing$. Since $l \leq i$, the claim follows, as
   $\YN{x_j}{l}$, the $l$-neighborhood of $x_j$, is contained in $\YN{x_j}{i}$.
\end{proof}



We now proceed to the actual argument for server availability, beginning with some notation. For $x \in X$, let $\A{x}{i}$
denote the set of available servers within $\YN{x}{i} = \{\yd{j}{x} \ | 1
\leq j \leq i\}$ at the {\em beginning} of iteration $i$. Thus, $|\A{x}{k}| =
k$. Furthermore, $\A{x}{i-1} \subseteq \A{x}{i}$ for $l + 1 \leq i \leq k$.
Obviously, $\A{x}{i} \subseteq \YN{x}{i}$. 

The {\em threshold level} of a net client $x$ (denoted by $\thr(x)$) is
defined as:
\[
  \forall x \in \bigcup\limits_{i = l}^k X_i, \qquad \thr(x) = 
     \begin {cases}
        k, & \mbox{if } x \in X_k \\
        j, & \mbox{if } x \in X_j \setminus X_{j+1}, \quad l \leq j < k 
     \end {cases}
\]
The threshold level of $x$ denotes the iteration of the outer loop of the
algorithm in which client $x$ first enters the net. In any iteration $k \geq j \geq \thr(x) + 1$, the client $x$
can lose neighboring servers because of the server choices made by (the
algorithm for) other clients, i.e., clients in the net $X_j$. On the other hand, for $l+1 \leq j \leq \thr(x)$, $x$ is itself part
of the net $X_j$. In these iterations, it can only lose neighboring
servers because of its own server choices. The next two claims address these two
phases.  

We now show that any net client $x$ has enough available
servers in its $\thr(x)$ neighborhood at the iteration $i = \thr(x)$ of the
outer loop of~\Cref{alg:HC}.

\begin{claim}
\label{cl:loop-invariant}
Let $x$ be any net client, and let $i = \thr(x)$. Then 
\begin{clenum}
    \item $|\A{x}{i} \cap \YN{x}{l}| \geq l - (k-i)$. \label{cl:loop-invariant.1}
    \item $|\A{x}{i}| \geq 2i - k = k - 2(k-i)$. \label{cl:loop-invariant.2}
\end{clenum}
\end{claim}  
\begin{proof}
    We look at the servers chosen during iteration $j$ of the outer loop, for 
    $i < j \leq k$. Note that $x$ didn't belong to the net $X_j$. Consider any
    client $x_j \in X_j$. By~\Cref{cl:subset}, $\YN{x}{i}$ does not intersect
    the $l$-neighborhood ball $\YN{x_j}{l}$. Hence, during the execution
    of~\Cref{lin:priv-server} in the inner for loop corresponding to $x_j$, no
    server is made unavailable from $\YN{x}{i}$. This is because the server
    chosen in~\Cref{lin:priv-server} belongs to $\YN{x_j}{l}$.

    Thus, during iteration $j$, servers from $\YN{x}{i}$ can become unavailable
    only during the execution of~\Cref{lin:shared-server} of the inner for loop.
    We note that by~\Cref{cl:chosen-unchosen}, there is at most one client
    $x_j \in X_j$ such that $\YN{x_j}{j} \cap \YN{x}{j} \neq \varnothing$.
    Thus, at most one server from  $\YN{x}{i}$ is made unavailable in iteration $j$. 

    We conclude that across the $k - i$ iterations before iteration $i$, there
    can be at most $k - i$ servers from
    $\YN{x}{i}$ that have been made unavailable.  Hence, $|\A{x}{i}| \geq i - (k - i) = 2i - k$.  Since
    $|\YN{x}{i} \cap \YN{x}{l}| = l$ and at most $k - i$ servers are made
    unavailable from the $i$-neighborhood ball $\YN{x}{i}$, $|\A{x}{i} \cap
    \YN{x}{l}| \geq l - (k - i)$.
\end{proof}

For any net client $x$, \Cref{cl:loop-invariant} shows that in iteration $i =
\thr(x)$, when $x$ first enters the net, there are enough available
servers in $\YN{x}{i}$. The following claim aids in asserting this for
subsequent iterations, by arguing that in any iteration $i \leq \thr(x)$, at
most $2$ available servers are made unavailable from $\YN{x}{i}$. 

\begin{claim}
\label{cl:atmost-two}
Let $x$ be any net client and $l + 1 \leq i \leq \thr(x)$. Then 
\begin{clenum}
    \item $|\A{x}{i - 1}| \geq |\A{x}{i}| - 2$
        \label{cl:atmost-two.1}
    \item If $|\A{x}{i - 1}| = |\A{x}{i}| - 2$, then one of the servers 
        in $\A{x}{i} \setminus \A{x}{i-1}$ is the farthest server in $\A{x}{i}$
        from $x$. \label{cl:atmost-two.2}
\end{clenum}
\end{claim}
\begin{proof}
    Note that $x \in X_i$ since $i \leq \thr(x)$. Consider any $x_c \in X_i \setminus \{x\}$. In the iteration
    of the inner for loop (\Cref{lin:inner-for-start}) corresponding to $x_c$,
    any servers that are made unavailable belong to $\YN{x_c}{i}$ and are
    therefore not in $\YN{x}{i}$, by \Cref{cl:chosen-chosen} (since $x,x_c \in
    X_i$). Thus, if any servers in $\A{x}{i} \subseteq
    \YN{x}{i}$ become unavailable in iteration $i$, then this can happen only in
    the iteration of the inner for loop corresponding to $x$. In this iteration
    of the inner for loop, the servers that become unavailable are $y_s$, the
    farthest server from $x$ in $\A{x}{i}$, and $y_p$, a different server that
    is chosen from the available servers in $\YN{x}{l}$. Note that $\{y_p, y_s\}
    \subseteq \A{x}{i}$. Thus, only the two servers $y_s, y_p$ in $\A{x}{i}$
    become unavailable in iteration $i$. Furthermore, if $\yd{i}{x} \in
    \A{x}{i}$ then $\yd{i}{x}$ is the farthest server in $\A{x}{i}$ from $x$,
    and thus $\yd{i}{x} = y_s$. Thus $\A{x}{i} \setminus \A{x}{i-1} = \{y_s,
    y_p\}$, and~\Cref{cl:atmost-two.1} holds. Since $y_s$ is the farthest
    server in $\A{x}{i}$,~\Cref{cl:atmost-two.2} holds as well.
\end{proof}

The following two claims show that our algorithm always succeeds in finding available
servers.
\begin{restatable}{claim}{available}
    \label{cl:available}
    For any $l + 1 \leq i \leq k$, and any $x_c \in X_i$: 
    \begin{clenum}
        \item There is an available server in $\YN{x_c}{i}$ when the algorithm
            executes \Cref{lin:shared-server} in the iteration of the inner for
            loop (\Cref{lin:inner-for-start}) corresponding to $x_c$. 
            \label{cl:available.1}
        \item There is an available server in $\YN{x_c}{l}$ when the algorithm
            executes \Cref{lin:priv-server} in the iteration of the inner for
            loop (\Cref{lin:inner-for-start}) corresponding to $x_c$.
            \label{cl:available.2}
    \end{clenum}
\end{restatable}

The proof of this claim, which follows from the previous claims, is given in
\Cref{sec:server-avail-claim}.

If $k$ is even, the algorithm does not look for available servers in iteration $i =
l$. If $k$ is odd, the algorithm will look for available servers in iteration
$i = l$, in \Cref{lin:priv-server}. The following claim extends the previous one
to handle this. The proof is a straightforward extension of the proof of the previous claim, and is therefore omitted. 

\begin{claim}
Suppose $k$ is odd. For iteration $i = l$, and any $x_c \in X_i$, there is
an available server in $\YN{x_c}{l}$ when the algorithm executes
\Cref{lin:priv-server} in the iteration of the inner for loop
(\Cref{lin:inner-for-start}) corresponding to $x_c$.
\label{cl:available-2}
\end{claim}
 
This completes the proof of \Cref{lem:spec-full}.

\section{Solving The MMC Problem}
\label{sec:algo}
In this section, we present \Cref{alg:MMC}, a constant factor approximation for the MMC
problem. Recall that our input consists of two point sets $Y$ (servers) and
$X$( clients) in an arbitrary metric space $(X \cup Y, d)$, a positive integer
$k$ that represents the coverage demand of each client, and the constant $\alpha
\geq 1$.

Our algorithm first computes a family $\F$ consisting of $k$ pairwise disjoint
subsets of $Y$, using the algorithm of \Cref{lem:spec-full}. It then invokes
$\text{Cover}(X,Y',\alpha)$, for each $Y' \in \F$, to compute a near-optimal
$1$-cover of $X$ using only the servers in $Y'$. Since there are $k$ server
subsets in $\F$, we obtain $k$ 1-covers of $X$. The algorithm then returns $r$, the
union of the $k$ covers. Because server subsets in $\F$ are disjoint, this union
yields a $k$-cover of $X$.

\begin{algorithm*}[hbt]
 \caption{\mmc}
 \label{alg:MMC}
 \begin{algorithmic}[1]
 \setcounter{ALC@unique}{0}
    \STATE For each $y \in Y$, assign $r(y) \leftarrow 0$.
    \STATE $\F \leftarrow \spec$.
    \FORALL {$Y' \in \F$} 
       \STATE $\bar{r} \leftarrow \text{Cover}(X, Y', \alpha)$.
       \STATE Let $r(y') \leftarrow \bar{r}(y')$ for each $y' \in Y'$.
    \ENDFOR
    \RETURN The assignment $r: Y \rightarrow \Real^+$.
 \end{algorithmic}
\end{algorithm*}

\subsection{Approximation Guarantee}
\label{sec:ag}
Note that \Cref{alg:MMC} computes the family $\F = \{Y_k^s, Y_k^p, Y_{k-1}^s, Y_{k-1}^p, \ldots \}$ as detailed in \Cref{lem:spec-full}. Let $Y_i \in \F$ be one such subset, where $Y_i$ may be either $Y_i^p$ or $Y_i^s$. $Y_i$ has the property that for any $x \in X$, there is an $x' \in X$ that is within two hops of $x$ in $G_i$ such that $\YN{x'}{i} \cap Y_i \neq \varnothing$. The following claim uses this property to argue that there is an inexpensive $1$-cover of $X$ that only uses servers from $Y_i$. The $1$-cover is constructed by using the servers in $Y_i$ to ``host'' the balls in the outer cover $\rho_i$.

\begin{claim} 
   Assume that either (a) $l + 1 \leq i \leq k$ and $Y_i$ is either $Y_i^p$ or $Y_i^s$, or (b) $k$ is odd, $i = l$, and $Y_i = Y_i^p$. Let $\rho_i$ be any 
outer cover of level $i$ for $X$ using servers from $Y$. There is a $1$-cover of $X$ that uses servers from $Y_i$ and has cost at most $12^{\alpha} \cdot \cost(\rho_i)$.
\label{cl:Bi}
\label{cl:cost-lambda}
\end{claim}

\begin{proof}
Consider the set $B$ of balls obtained by expanding each ball in the outer cover
$\rho_i$ to $6$ times its original radius. We claim
\begin{claim}
For any client $x \in X$, there is some ball in $B$ that contains $x$ as well as at least one server in $Y_i$.
\label{cl:assert}
\end{claim}

Before proving \Cref{cl:assert}, we first prove \Cref{cl:Bi} using it. We construct a set $B'$ of balls as follows. Consider any ball $b \in B$. If it does not contain a server from $Y_i$, we ignore it. If it does contain a server in $Y_i$, pick an arbitrary such server $y$, translate $b$ so that it is centered at $y$, double its radius, and add the resulting ball to $B'$. 

It is possible at this stage that for a server $y \in Y_i$, there are several balls in $B'$ centered at $y$. From each such concentric family, discard from $B'$ all but the largest of the concentric balls.  It follows from \Cref{cl:assert} that $B'$ covers each client in $X$. Since each ball in $B'$ is obtained by translating and scaling some ball in the outer cover $\rho_i$ by a factor of $12$, the cost of $B'$ is at most
$12^{\alpha} \cdot \cost(\rho_i)$. This establishes \Cref{cl:Bi}.

We now turn to the proof of \Cref{cl:assert}. From the definition of $G_i$, we have that for any edge $(x',x'')$ in $G_i$, 
\begin{equation}
    \label[ineq]{ineq:edge-in-Gi}
    d(x',x'') \leq d(x', \yd{i}{x'}) + d(x'',\yd{i}{x''}). 
\end{equation} 

Now consider an arbitrary client $x \in X$. By \Cref{lem:spec-full}, there is a path $\pi$ in $G_i$ with at most $2$ edges (and $3$ vertices) that connects $x$ to some vertex $\bar{x}$, with $\YN{\bar{x}}{i} \cap Y_i \neq \varnothing$.

 Let $\delta(y, \rho_i(y))$ be the largest ball in outer cover $\rho_i$ that
 serves at least one vertex on path $\pi$. Suppose that it serves vertex
 $\hat{x} \in \pi$. ($\hat{x}$ could be the same as $x$ or $\bar{x}$.)
 See~\Cref{fig:factor12} for an illustration. Using the definition of an outer cover of level $i$,
 and the way we pick the ball $\delta(y, \rho_i(y))$, it follows that for any
 vertex $x' \in \pi$,
\begin{equation}
\label[ineq]{ineq:rhi}
d(x', \yd{i}{x'}) \leq \rho_i(y). 
\end{equation}

\begin{figure}[ht]
\centering
\includegraphics[scale=0.7]{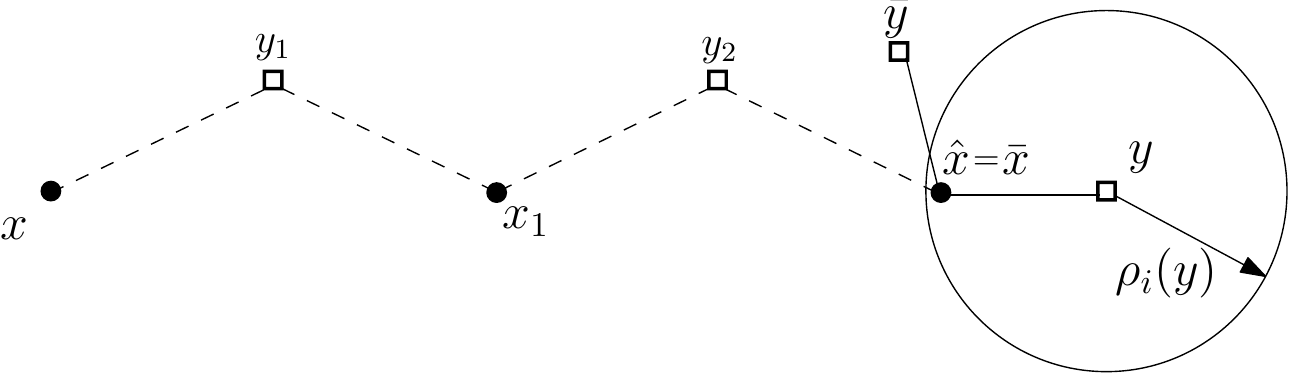}
\caption{Illustration for the proof of~\Cref{cl:assert}. For the client $x \in X$,
   the dashed edges correspond to a path $\pi$ in $G_i$ from $x$ to $\bar{x}$.
   Here, $y_1 \in N_i(x) \cap N_i(x_1)$, $y_2 \in N_i(x_1) \cap N_i(\bar x)$,
   and $\bar{y} \in Y_i \cap N_i(\bar{x})$. The ball $\delta(y, \rho_i(y))$
   serves $\hat{x}$. Here, $\hat{x}$ happens to be $\bar{x}$. Note that we can get
   from $y$ to $x$ using $5$ edges of the figure, and from $y$ to $\bar{y}$
   using $2$ edges. Therefore, expanding the ball at $y$ by a factor of $6$ will
   cover both $x$ and $\bar{y}$. (In this example, even a factor of $5$
   suffices.) }
\label{fig:factor12}
\end{figure}

Thus, 
\[ d(y,x) \leq d(y,\hat{x}) + \left( \sum\limits_{(x',x'') \in \pi[\hat{x},x]} d(x',x'') \right)  \leq 5 \rho_i(y).\]
Here, we denote by $\pi[\hat{x},x]$ the sub-path of $\pi$ from $\hat{x}$ to $x$, and use Inequalities~\ref{ineq:edge-in-Gi} and~\ref{ineq:rhi} in the second step.

Now, $\YN{\bar{x}}{i} \cap Y_i \neq \varnothing$. Let $\bar{y} \in \YN{\bar{x}}{i} \cap Y_i$ be chosen arbitrarily. Clearly, $d(\bar{x},\bar{y}) \leq d(\bar{x}, \yd{i}{\bar{x}}) \leq \rho_i(y).$

We calculate that 
\[ d(y,\bar{y}) \leq d(y,\hat{x}) + \left( \sum\limits_{(x',x'') \in \pi[\hat{x},\bar{x}]} d(x',x'') \right) + d(\bar{x},\bar{y}) \leq 6 \rho_i(y).\]

Thus, the ball $\delta(y, 6 \rho_i(y))$ contains both $x$ and $\bar{y} \in Y_i$, completing the proof of \Cref{cl:assert}.

\end{proof}
\begin{remark} 
With a more detailed argument, the factor $12^{\alpha}$ can be improved. For instance, a bound of  $11^{\alpha}$ is almost immediate from the proof.
\end{remark}

We can now establish the approxmation guarantee for \Cref{alg:MMC} and our main result.
 
\begin{theorem}
    \label{thm:result}
    Given point sets $X$ and $Y$ in a metric space $(X \cup Y, d)$ and a 
    positive integer $k \leq |Y|$, \Cref{alg:MMC} runs in polynomial 
    time and returns a $k$-cover of $X$ with cost at most $2 \cdot {(12 \cdot 9)}^{\alpha}$ times that of an optimal $k$-cover.
\end{theorem}

\begin{proof}
It is evident that the algorithm runs in polynomial time, and we have already noted that it returns a $k$-cover $r$. Let $r'$ be any optimal assignment. By \Cref{lem:lbound}, there exist outer covers $\rho_i$, for $1 \leq i \leq k$, such that 
\[\sum\limits_{i = 1}^k  \cost(\rho_i) \leq 3^{\alpha} \cdot \cost(r').\]

Assume that either (a) $l + 1 \leq i \leq k$ and $Y_i$ is either or $Y_i^p$ or $Y_i^s$, or (b) $k$ is odd, $i = l$, and $Y_i = Y_i^p$. From \Cref{cl:Bi}, we conclude that there is a $1$-cover for $X$ that uses servers $Y_i$ and has cost at
most $12^{\alpha} \cdot \cost(\rho_i)$. Since $\Cover(X,Y_i,\alpha)$, which is invoked in \Cref{alg:MMC}  returns 
a $3^{\alpha}$ approximation, the cost of the $1$-cover it returns is at most
$(12 \cdot 3)^{\alpha} \cost(\rho_i)$. 

At most two $1$-covers are computed for each $i$, once with server set $Y_i^s$ and once with server set $Y_i^p$. Thus,
 
\[
    \cost(r) 
    \leq 2 \cdot (12 \cdot 3)^{\alpha} \cdot \sum\limits_{i = l}^k \cost(\rho_i)
    \leq 2 \cdot (12 \cdot 3)^{\alpha} \cdot \sum\limits_{i = 1}^k \cost(\rho_i)
         \leq 2 \cdot {(12 \cdot 9)}^{\alpha} \cdot \cost(r').
\]
\end{proof}

\section{The \texorpdfstring{$t$}{t}-MMC Problem}
\label{sec:tMMC}
In this section, we describe a natural generalization of the MMC problem, called the $t$-MMC problem. The input to this problem is similar to the MMC problem -- the two point sets $Y$ (servers) and $X$ (clients) in an arbitrary metric space $(X \cup Y, d)$, a positive integer $k$ that represents the coverage demand of each client, a constant $\alpha$. There is an additional input, an integer $t$, that represents the upper bound on the number servers that can be opened or \emph{used} in the solution.

A \emph{$k$-cover using at most $t$ servers} is a subset $Y' \subseteq Y$ such that $|Y'| \le t$, together with an assignment $r: Y' \to \mathbb{R}^+$ that $k$-covers $X$. Here, the cost of the solution is defined as $\cost(r) = \sum_{y \in Y'} (r(y))^\alpha$. Intuitively, the restriction of $t$ is analogous to the cardinality restrictions imposed on the solutions in problems like $t$-center, $t$-median and so on.

Now, the goal of the $t$-MMC problem is to compute a minimum cost $k$-cover using at most $t$ servers. In comparison to the MMC problem, the additional complexity arises from having to decide which $t$ servers to use for $k$-covering $X$. We give an $O(1)$ approximation for this problem. Here, we assume that $k \le |Y|$ and $k \le t$, so that the given instance is feasible. 

\subsection{Algorithm}
\label{subsec:tMMC-alg}
The $O(1)$ approximation algorithm for the $t$-MMC problem consists of the following steps.
\begin{enumerate}
\item We first compute a family $\F = \{Y_k^s, Y_k^p, Y_{k-1}^s, Y_{k-1}^p, \ldots \}$ consisting of $k$ pairwise disjoint subsets of $Y$, using the algorithm of \Cref{lem:spec-full}. For convenience, let us rename this family of servers as $\F = \{V_1, V_2, \cdots, V_k\}$ respectively.
\item For each $1 \le i \le k$, and for each $1 \le t_i \le t$, we compute a $1$-cover of $X$, using at most $t_i$ servers from $V_i$. Here, we use the polynomial time approximation algorithm of Charikar and Panigrahy \cite{CharikarP04} for computing $1$-cover using at most $t_i$ servers. Let us denote the solution returned by their algorithm by $S(V_i, t_i)$.
Even though their algorithm is stated for the case of $\alpha = 1$, it generalizes to any $\alpha \ge 1$. It can be shown that the approximation guarantee of their algorithm is $5^\alpha$. 

\item Let us call a $k$-tuple $(t_1, t_2, \ldots, t_k)$ a \emph{valid $k$-tuple} if $1 \le t_i \le t$ for each $i$, and $\sum_{i = 1}^k t_i \le t$.

We compute a valid $k$-tuple $(t_1^*, t_2^*, \ldots, t_k^*)$ that minimizes $\sum_{i = 1}^k \cost(S(V_i, t_i))$, over all valid $k$-tuples $(t_1, t_2, \cdots, t_k)$. Such a valid $k$-tuple can be computed in polynomial time using dynamic programming. We return $\bigcup_{i = 1}^k S(V_i, t_i^*)$ as our solution.
\end{enumerate}

\subsection{Approximation Guarantee}
\label{subsec:tMMC-approx}
It is easy to see that the algorithm described above runs in polynomial time. Also, for each $1 \le i \le k$, $S(V_i, t_i^*)$ $1$-covers $X$ using disjoint servers. Since the final solution $\bigcup_{i = 1}^k S(V_i, t_i^*)$ obtained using dynamic programming is a valid $k$-tuple,  the algorithm computes a $k$-cover of $X$ that uses at most $t$ servers.

For proving the approximation guarantee, we extract from the optimal solution to the $t$-MMC problem, the outer covers $\rho_i$ for each $1 \le i \le k$ with some special properties. The following is an analogue of \Cref{lem:lbound}, however some new ideas are needed to handle the restriction on the number of servers that can be used in the resultant outer covers. The proof of the following theorem is given in \Cref{subsec:outercover-for-tMMC}.

\begin{restatable}{theorem}{toutercover}
\label{thm:tMMC-outercovers}
Let $r': Y' \to \mathbb{R}^+$ be an assignment that constitutes an optimal solution to the $t$-MMC problem, where $Y' \subseteq Y$ with $|Y'| \le t$. For each $l \le i \le k$, we can find level $i$ outer cover $\rho_i$ that uses $t_i'$ servers, such that
\begin{itemize}
\item If $k$ is even, then\ \vspace{0.1cm}\\
\begin{tabular}{cc}

\textbf{1.} $\sum_{i = l + 1}^k \cost(\rho_i) \le 2 \cdot (3 \cdot 3)^\alpha \cdot \cost(r')$, and\qquad  &\textbf{2. } $\sum_{i = l + 1}^k 2 \cdot t_i' \le t$. \\ 
\end{tabular}\vspace{0.1cm}
 
\item If $k$ is odd, then\ \vspace{0.1cm}\\
\begin{tabular}{cc}

\textbf{1.} $\sum_{i = l\ \ }^k \cost(\rho_i) \le 2 \cdot (3 \cdot 3)^\alpha \cdot \cost(r')$, and\qquad &\textbf{2. } $t_l' + \sum_{i = l + 1}^k 2 \cdot t_i' \le t$. \\ 
\end{tabular}

\end{itemize}
\end{restatable}

Given an outer cover $\rho_i$ that uses at most $t'_i$ servers, the following claim constructs an inexpensive $1$-cover of $X$ using at most $t'_i$ servers from $Y_i$. This will help us bound the cost of the solution returned by the algorithm from \Cref{subsec:tMMC-alg}. This claim strengthens \Cref{cl:cost-lambda}, but the proof generalizes easily. 


\begin{claim}
\label{claim:cost-size-lambda}
Assume that either (a) $l + 1 \le i \le k$ and $Y_i$ is either $Y_i^p$ or $Y_i^s$, or (b) $k$ is odd, $i = l$ and $Y_i = Y_i^p$. Let $\rho_i$ be an outer cover of level $i$ using at most $t_i'$ servers from servers from $Y$. Then there is a $1$-cover of $X$ that uses at most $t_i'$ servers from $Y_i$, and has cost at most $12^\alpha \cdot \cost(\rho_i)$ .
\end{claim}

Now, we establish the approximation guarantee for the algorithm described in \Cref{subsec:tMMC-alg}.

\begin{theorem}
\label{thm:tMMC-main}
Given point sets $X$ and $Y$ is a metric space $(X \cup Y, d)$, and positive integers $k$ and $t$ such that $k \le |Y|$ and $k \le t$, the algorithm described in \Cref{subsec:tMMC-alg} runs in polynomial time, and returns a $k$-cover of $X$ using at most $t$ servers from $Y$, and with cost at most $4 \cdot (540)^\alpha$ times that of an optimal $k$-cover that uses at most $t$ servers from $Y$.
\end{theorem}
\begin{proof}
We focus on the case where $k$ is even. The case where $k$ is odd is similar, and is therefore omitted.

We have already argued that the algorithm runs in polynomial time, and the solution produced by the algorithm $k$-covers $X$ using at most $t$ servers. 

Let $r': Y' \to \mathbb{R}^+$ be any optimal assignment that $k$-covers $X$, where $Y' \subseteq Y$, with $|Y'| \le t$. By \Cref{thm:tMMC-outercovers}, there exist outer covers $\rho_i$ that use $t_i'$ servers such that $\sum_{i = l+1}^k \cost(\rho_i) \le 2 \cdot (3 \cdot 3)^\alpha \cdot \cost(r')$, with $\sum_{i = l+1}^k 2 \cdot t_i' \le t$.

For each of $(Y_i^p, t_i', \rho_i)$ and $(Y_i^s, t_i', \rho_i)$, we use \Cref{claim:cost-size-lambda}, to argue that there exist two $1$-covers from $Y_i^p$ and $Y_i^s$ respectively. These $1$-covers have cost at most $(12)^\alpha \cdot \cost(\rho_i)$ each, and each uses at most $t_i'$ servers. Since the $5^\alpha$ approximation of Charikar and Panigrahy \cite{CharikarP04} is used to get two $1$-covers $S(Y_i^p, t_i')$ and $S(Y_i^s, t_i')$, we have that $\cost(S(Y_i^p, t_i')) \le (12 \cdot 5)^\alpha \cdot \cost(\rho_i)$ and $\cost(S(Y_i^s, t_i')) \le (12 \cdot 5)^\alpha \cdot \cost(\rho_i)$.

Note however that $\sum_{i = l+1}^k 2\cdot t_i' \le t$, so $(t_k', t_k', \cdots, t_l', t_l')$ is a valid $k$-tuple, and so the dynamic program of step 3 must have considered the solution
$$\left(\bigcup_{i = l+1}^k S(Y_i^p, t_i')\right) \cup \left(\bigcup_{i = l+1}^k S(Y_i^s, t_i')\right).$$
Since the cost of the solution output by the dynamic program is at most the cost of solution corresponding to this tuple, we have that,
\begin{align*}
\sum_{i = 1}^k \cost(S(V_i, t_i^*)) &\le \sum_{i = l+1}^k \left(\cost(S(Y_i^p, t_i')) + \cost(S(Y_i^s, t_i'))\ \right)
\\&\le 2 \cdot (12 \cdot 5)^\alpha \cdot \sum_{i = l+1}^k  \cost(\rho_i)
\\&\le 2 \cdot (12 \cdot 5)^\alpha \cdot 2 \cdot 9^\alpha \cdot \cost(r') = 4 \cdot (540)^\alpha \cdot \cost(r').
\end{align*} 
\end{proof}

\bibliographystyle{plainnat}
\bibliography{mmc_ref}
\clearpage
\markboth{Appendix}{}

\appendix
\section{The Outer Cover Lower Bounds}
\label{sec:outer-cover}
In this section, we prove the lower bounds on the optimal solutions of the MMC problem and the $t$-MMC problem respectively.  
\subsection{The Outer Cover Lower Bound for the MMC Problem}
\label{subsec:outercover-for-MMC}
In this section, we provide the proof, adapted from \cite{BhowmickVX15}, of
\Cref{lem:lbound}. For convenience, we restate the theorem.

\outercover*
\begin{proof}
    Let $B = \{\disk(y, r'(y)) \mid y \in Y\}$ denote the set of balls
    corresponding to the assignment $r'$. We show that it is possible to
    form subsets $B_i \subseteq B$, for 
    each $1 \leq i \leq k$ such that:
    \begin{enumerate}
        \item $\mu_i \leq 3^{\alpha} \cdot \cost(B_i)$.
        \item $B_i \cap B_j = \varnothing$, for each $1 \leq i \neq j \leq k$.
        \item No two balls in $B_i$ intersect, for each $1 \leq i \leq k$.
    \end{enumerate}

    If we show this,~\Cref{lem:lbound} follows because
    \[
        \sum\limits_{i = 1}^k\mu_i \leq 3^{\alpha} \cdot \sum\limits_{i = 1}^k\cost(B_i)
                                  \leq 3^{\alpha} \cdot \cost(B) 
                                  = 3^{\alpha} \cdot  \cost(r').
    \]


    We create the set of balls $B_i$ in a
    top-down manner as described in \Cref{alg:analysis}.

    \begin{algorithm*}[hbt]
     \caption{Compute-Balls}
     \label{alg:analysis}
     \begin{algorithmic}[1]
     \setcounter{ALC@unique}{0}
        \REQUIRE The set of balls $B$ corresponding to a $k$-cover assignment $r'$
        \ENSURE The set of balls $B_i, 1 \leq i \leq k$.
        \FOR {$i = k$ to $1$}
          \STATE Let $\lgt_i(x) \leftarrow $ The largest ball in $B$ that contains $x$.
          \STATE Let ${B_i}' = \{\lgt_i(x) \mid x \in X\}$. \label{lin:Bi}
          \STATE $B_i \leftarrow \varnothing$.
          \WHILE {${B_i}' \neq \varnothing$}
            \STATE Let $b$ be the largest ball in ${B_i}'$.
            \STATE $N \leftarrow $ Set of balls in ${B_i}'$ that intersect $b$.
            \COMMENT {Note: $b \in N$.}
            \STATE $B_i \leftarrow B_i \cup \{b\}$.
            \STATE ${B_i}' \leftarrow {B_i}' \setminus N$.
          \ENDWHILE
          \STATE $B \leftarrow B \setminus B_i$.
        \ENDFOR
     \end{algorithmic}
    \end{algorithm*} 

    We thus have a set of balls $B_i, 1 \leq i \leq k$. It is clear that 
    $B_i \cap B_j = \varnothing$ (Property 2), and no two balls in $B_i$ intersect (Property 3). 
    
    We now verify that each $B_i$ also satisfies Property $1$. For this, consider
    $L_i$, the set of balls obtained by increasing the radius of each ball in
    $B_i$ by a factor of $3$. We argue that $L_i$ is an outer cover of level $i$ for $X$. 

    Fix $x \in X$, and consider the ball $\lgt_i(x)$ in Line 3 of
    iteration $i$. At this point, the balls in $\bigcup_{j= i + 1}^k B_j$ have
    been removed from the original $B$, which had at least $k$ balls containing $x$. Since no two balls in $B_j$ intersect, there is at most one ball in each $B_j$ that contains $x$. Thus, at this point, there are at least $i$ balls left in $B$ that contain $x$. Thus, the radius of
    $\lgt_i(x)$ is at least $d(x, \yd{i}{x})$.
    \begin{enumerate}
       \item If $\lgt_i(x) \in B_i$, then the corresponding ball in $L_i$ has
           radius at least $d(x, \yd{i}{x})$.
       \item If $\lgt_i(x) \notin B_i$, then there is an even larger ball $b$ in $B_i$
           that intersects $\lgt_i(x)$. The ball obtained by multiplying the radius of $b$ by $3$ is in $L_i$; it contains $\lgt_i(x)$ and thus $x$; and it
           has radius at least $d(x, \yd{i}{x})$.
    \end{enumerate}

    Thus, $L_i$ is an outer cover of level $i$ for $X$. We infer that 
    \[
        \mu_i 
        \leq \cost(L_i) 
        \leq 3^{\alpha} \cdot \cost(B_i).
    \]

    Thus, Property~1 holds.
\end{proof}

\subsection{The Outer Cover Lower Bound for the \texorpdfstring{$t$}{t}-MMC Problem}
\label{subsec:outercover-for-tMMC}
In this section, we prove the \Cref{thm:tMMC-outercovers}, which generalizes \Cref{lem:lbound} in the case when the size of the outer covers are restricted to satisfy certain properties. For convenience, we restate the theorem.
\toutercover*
\begin{proof} For simplicity, we prove the theorem only for the case where $k$ is even. The proof of the case where $k$ is odd is similar, and is therefore omitted.

Note that any feasible solution to the $t$-MMC problem is also feasible for the MMC problem. Therefore, we can use \Cref{lem:lbound} to extract from the assignment $r'$, the outer covers $\bar{\rho_i}$ for each $1 \le i \le k$, such that $\sum_{i = 1}^k \cost(\bar{\rho_i}) \le 3^\alpha \cdot \cost(r')$.

These outer covers satisfy the first property, but they may not satisfy the second property. However, if the outer cover $\bar{\rho_i}$ uses $\bar{t_i}$ servers, then it is easy to verify that the proof of \Cref{lem:lbound} ensures that $\sum_{i = 1}^k \bar{t_i} \le t$.

Let us order the above outer covers $\bar{\rho_i}$ in a nondecreasing order of the number of servers used, and rename them according to this ordering as $r_k, r_{k-1}, \ldots, r_1$. Let $t_i'$ denote the number of servers used by the outer cover $r_i$. To transform the outer covers $\bar{\rho_i}$ into the outer covers $\rho_i$ that satisfy the second property of the theorem, we need the following claim.

\begin{claim}
\label{claim:bounded-outercover}
Let $\bar{\rho_i}$ be an outer cover of level $i$, $l \le i \le k$, and $r$ be any $1$-cover that uses at most $t'$ servers. Then there is an outer cover $\rho_i$ of level $i$ that uses at most $t'$ servers, and $\cost(\rho_i) \le 3^\alpha \cdot (\cost(\bar{\rho_i}) + \cost(r))$.
\end{claim}
\begin{proof}
Let $B$ and $B'$ be the set of balls corresponding to the outer cover $\bar{\rho_i}$ and the $1$-cover $r$ respectively. We describe an iterative procedure to compute a set $R_i$ of balls, which is initially empty. Each ball in $B$ and $B'$ is initially ``unmarked''. Pick any unmarked ball $b_j' \in B'$, and suppose $b_j' = \delta(y_j', r_{b_j'})$. Let $B_j$ denote the set unmarked of balls from $B$ that serve a client $x \in X$ that is also covered by $b_j'$. Let $r_j'$ be the maximum radius from the set of balls $B_j \cup \{b_j'\}$. Add the ball $\delta(y_j', 3r_j')$ to the set $R_i$. Mark the ball $b_j'$ from $B'$ and the balls $B_j$ from $B$, and repeat the above process until all balls from $B'$ are marked. 

We argue that at the end of this process, the radius assignment $\rho_i$ corresponding to $R_i$ is an outer cover of level $i$ with the claimed properties. Without loss of generality, we assume that each ball $b \in B$ serves some client $x \in X$. Consider a client $x \in X$, and a ball $b \in B$ that serves it. Since $r$ is a valid $1$-cover, there exists a ball $b' \in B'$ that also covers $x$. So, if $b$ was not considered in any iteration before, it will be considered in the iteration when $b'$ is marked, and both $b$ and $b'$ will be marked by the end of that iteration. Thus, at the end of the above process, all balls $b \in B$ will be marked. Also, for each ball $b' \in B'$, we add exactly one ball to $R_i$. Therefore, $|R_i| = |B'| \le t'$.

Now we argue that $\rho_i$ is an outer cover of level $i$. Consider any client $x \in X$. Since $\bar{\rho_i}$ is an outer cover of level $i$, there exists a ball $b \in B$ centered at some $y \in Y$ with radius $r_b \ge d(x, y_i(x))$ that covers $x$. Using the argument from the above paragraph, such a ball $b$ was marked in some iteration. Suppose the ball from the set $B'$ that was marked in that iteration was $b_j' = \delta(y_j', r_{b_j}')$, and the ball $\delta(y_j', 3 r_j')$ was added to $R_i$. Now, 
$$d(x, y_j') \le \ d(x, y) + d(y, x') + d(x', y_j') \le \ r_b + r_b + r_{b_j'} \le 3r_j'$$
Here, $x' \in X$ is a common client served by $b$ in $\bar{\rho_i}$ and covered by $b'$ in $r$. Note that $x'$ may or may not be same as $x$. The last inequality follows because of the choice of $r_j'$. Therefore, the ball $\delta(y_j', 3r_j')$ covers $x$. Also since $d(x, y_i(x)) \le r_b \le r_j' \le 3r_j'$, the ball $\delta(y_j', 3r_j')$ also serves $x$. Thus, it follows that $\rho_i$ is an outer cover of level $i$.

The ball added to $R_i$ in a particular iteration has radius $3r_j'$, where $r_j'$ is the maximum radius from the set $B_j \cup \{b_j'\}$. Now, considering all such balls in $R_i$, the bound on the cost of $\rho_i$ follows.
\end{proof}

Now, we use \Cref{claim:bounded-outercover} for pairs of outer covers
$(\bar{\rho_i}, r_{i})$ to get an outer cover $\rho_i$ of level $i$, for each
$l+1 \le i \le k$ with the desired upper bounds on the cost and the number of
servers used. Now, 
\begin{align*}
\sum_{i = l+1}^k \cost(\rho_i) &\le \sum_{i = l+1}^k 3^\alpha \cdot (\cost(\bar{\rho_i}) + \cost(r_i)) 
\\&\le \sum_{i = 1}^k 3^\alpha \cdot (\cost(\bar{\rho_i}) + \cost(r_i)) 
\\&= 2 \cdot 3^\alpha \sum_{i = 1}^k \cost(\bar{\rho_i}) 
\\&\le 2 \cdot (3 \cdot 3)^\alpha \cost(r')
\end{align*}
The third inequality follows due to the fact that the set of outer covers $\{r_1, r_2, \cdots, r_k\}$ is same as the set of original outer covers $\{\bar{\rho_1}, \bar{\rho_2}, \cdots, \bar{\rho_k}\}$. Note that for each $l+1 \le i \le k$, $\rho_i$ uses at most $t_i'$ servers. Since the outer covers $r_i$ are ordered in a non-decreasing order of the number of servers used, $\sum_{i = l+1}^k 2\cdot t_i' \le \sum_{i = 1}^k t_i' \le t$, and the second property follows.
\end{proof}

\section{The Proof of \texorpdfstring{\Cref{cl:available}}{Claim 3.7}}
\label{sec:server-avail-claim}
In this section we give the proof of \Cref{cl:available} that shows that when \Cref{alg:HC} always succeeds in finding available servers.
\available*
\begin{proof}
    Since $x_c \in X_i$, we infer that $i \leq
    \thr(x_c)$. Using~\Cref{cl:atmost-two}, we have
    \begin{align*}
        |\A{x_c}{i}| & \geq |\A{x_c}{i + 1}| - 2 \\
                       & \geq |\A{x_c}{i + 2}| - 2 -2 \\
                       & \geq \dots \\
                       & \geq |\A{x_c}{\thr(x_c)}| - 2 \cdot (\thr(x_c) - i) \\
                       & \geq k - 2 (k - i)  \tag{$\because |\A{x_c}{\thr(x_c)}| \geq k - 2 (k - \thr(x_c))$} \\
                       & \geq 2     \tag{$\because i \geq l + 1$}
    \end{align*}
    Using an argument from the proof of~\Cref{cl:atmost-two}, none of the
    servers in $\A{x_c}{i}$ are made unavailable in iteration $i$ till $x_c$
    is considered in~\Cref{lin:inner-for-start}. Thus, there are at least two
    servers available when the algorithm executes~\Cref{lin:shared-server}
    corresponding to $x_c$, and~\Cref{cl:available.1} holds.

    The argument for~\Cref{cl:available.2} is similar but requires some case
    analysis. We begin by observing that when the algorithm executes
    \Cref{lin:priv-server} corresponding to $x_c$, there is at least one available
    server  $y \in \YN{x_c}{i}$. Now suppose that in some iteration $i+1 \leq j
    \leq \thr(x_c)$, $\A{x_c}{j} \setminus \A{x_c}{j-1}$ consists of two servers from
    $\YN{x_c}{l}$. By \Cref{cl:atmost-two.2}, a server from $\YN{x_c}{l}$ is the
    farthest server from $x_c$ in $\A{x_c}{j}$. This implies that all servers in
    $\A{x_c}{j-1}$ belong to $\YN{x_c}{l}$, and thus $y \in \YN{x_c}{l}$. This $y$
    is available when the algorithm  executes \Cref{lin:priv-server} corresponding
    to $x_c$. 
     
    We are left with the case that in each iteration $i+1 \leq j \leq
    \thr(x_c)$, $\A{x_c}{j} \setminus \A{x_c}{j-1}$ consists of at most
    one server from $\YN{x_c}{l}$. Using \Cref{cl:loop-invariant.1}, and the fact
    that $\thr(x_c) - i$ iterations have happened since iteration
    $\thr(x_c)$, we have 
    \[ 
       |\A{x_c}{i} \cap \YN{x_c}{l}| \geq
       |\A{x_c}{\thr(x_c)} \cap \YN{x_c}{l}| - (\thr(x_c) - i) \geq l - (k -
       i) \geq 1.
    \] 
    Thus there is at least one server $y' \in \A{x_c}{i} \cap
    \YN{x_c}{l}$. If the server chosen in \Cref{lin:shared-server} corresponding to
    $x_c$ belongs to $\YN{x_c}{l}$, then all available servers
    in $\YN{x_c}{i}$ belong to $\YN{x_c}{l}$. Thus, once again, some server in
    $\YN{x_c}{l}$ is available when the algorithm executes \Cref{lin:priv-server}
    corresponding to $x_c$. If the server chosen in \Cref{lin:shared-server}
    corresponding to $x_c$ does not belong to $\YN{x_c}{l}$, then $y' \in
    \YN{x_c}{l}$ is available when the algorithm  executes \Cref{lin:priv-server}
    corresponding to $x_c$. We have thus shown that~\Cref{cl:available.2}
    holds.
\end{proof}

\section{The Non-uniform MMC Problem}
\label{sec:conc}
In this section, we address the non-uniform version of the metric multi-cover problem, which we refer to as the {\em non-uniform MMC}, and present an $O(1)$ approximation for it. Recall that the input consists of two point sets $Y$(servers) and $X$(clients) in an arbitrary metric space $(X \cup Y, d)$, a constant $\alpha \geq 1$, and a coverage function $\kappa: X \rightarrow
\mathbb{Z}^+$. 

Consider an assignment $r: Y \rightarrow \Real^+$ of radii to each server
in $Y$. This can be viewed as specifying a ball of radius $r(y)$ at each server $y \in Y$. If, for each client $x \in X$, at least $\kappa(x)$ of the corresponding server balls contain $x$, then we say that $X$ is $\kappa$-covered by $r$. That is, $X$ is $\kappa$-covered if for each $x \in X$, 
\[ 
    | \{ y \in Y \ | \ d(x, y) \leq r(y) \} | \geq \kappa(x).
\]

Any assignment  $r: Y \rightarrow \Real^+$ that $\kappa$-covers $X$ is a 
feasible solution to the non-uniform MMC, and the goal is to find a feasible solution that minimizes the cost $\sum_{y
\in Y} \left ( r(y) \right )^{\alpha}$.  We assume that $\kappa(x) \leq |Y|$ for each client $x$, for
otherwise there is no feasible solution. 

To solve the non-uniform MMC problem, our plan is to partition the set of
servers $Y$ into disjoint sets and invoke a $1$-covering algorithm with each
server subset. Unlike the uniform case, each $1$-covering instance thus
generated may only cover a subset of the clients, and not all clients in $X$.
For example, a client $x$ such that $\kappa(x) = 100$ will be involved in $100$
$1$-covering instances, whereas a client $x'$ with demand $50$ would be in $50$
$1$-covering instances.

\subsection{Partitioning Servers}
\label{ssec:partition-n}
Our algorithm for partitioning $Y$ into server subsets uses a criterion that
generalizes that of~\Cref{lem:spec-full}. We adopt terminology for the
non-uniform case from~\Cref{sec:spec}. Let $k$ now denote  $\max_{x \in X}
\kappa(x)$.  For client $x \in X$, let its set of {\em private servers} be
$\YN{x}{l} = \YN{x}{\lceil\kappa(x)/2\rceil}$. For notational convenience, we
denote $\YN{x}{i}$, the set of $i$ nearest servers to $x$ in $Y$, by
$\NN{x}{i}$. 

Before stating the generalized lemma, we need some additional definitions. For
$1 \leq i \leq k$, we define the coverage function $\K{i}: X \rightarrow
\mathbb{Z}^+$ by $\K{i}(x) = \max \{0, \kappa(x) - (i-1)\}$. Thus, $\K{i}$ is
obtained by decreasing the original coverage requirement of each client by $i -
1$, with the proviso that we don't decrease below $0$. For each $1 \leq i \leq
k$, we define an undirected graph $\GL{i}$ with vertex set $X$. We add
$(x,x')$ as an edge in $\GL{i}$ if (a) $i \leq \lceil \kappa(x)/2 \rceil$; (b)
$i \leq \lceil \kappa(x')/2 \rceil$; and (c) $\NN{x}{\kappa(x) - (i-1)} \cap
\NN{x'}{\kappa(x') - (i-1)} \neq \varnothing$. Note that condition if (a) and
(b) hold, condition (c) can also be written as $\NN{x}{\K{i}(x)} \cap
\NN{x'}{\K{i}(x')} \neq \varnothing$. The conditions (a) and (b) ensure that a
client $x$ is isolated in graph $\GL{i}$ 
for $i > \lceil \kappa(x)/2 \rceil$. 

\begin{lem}
  \label{lem:general-interface}
  Let $l = \lceil k / 2 \rceil$. We can efficiently compute a family
  $\F$ of server subsets such that
  \begin{enumerate}
      \item \F contains two subsets \Ys{i} and \Yp{i} for each $1 \leq i < l$.
         For $i = l$, if $k$ is even, \F contains \Ys{l} and \Yp{l}, 
         else \F contains only \Yp{l}.
      \item \F is pairwise disjoint.
      \item Fix $1 \leq i \leq l$. 
      \begin{enumerate}
          \item For any client $x$ with $\kappa(x) \geq 2i - 1$, there is a
              client $x' \in X$ within $3$ hops of $x$ in \GL{i} such that
              $\Yp{i} \cap \NN{x'}{\K{i}(x')} \neq \varnothing$.
          \item If $k$ is even or $i < l$, for any client $x$ with $\kappa(x)
              \geq 2i$, there is a client $x' \in X$ within $3$ hops of $x$ in
              \GL{i} such that $\Ys{i} \cap \NN{x'}{\K{i}(x')} \neq
              \varnothing$.
      \end{enumerate}
  \end{enumerate}
\end{lem}

Going back to the non-uniform MMC problem, \Yp{i} will be used to $1$-cover the
clients $\{ x \in X \mid \kappa(x) \geq 2i - 1\}$, and \Ys{i} will be used to
$1$-cover $\{ x \in X \mid \kappa(x) \geq 2i \}$. Suppose that $\kappa(x_1) =
100, \kappa(x_2) = 50$ for some $x_1, x_2 \in X$. The plan is use each of the sets
\Ys{i}, \Yp{i} for $1 \leq i \leq 25$ to cover both $x_1$ and $x_2$ once. The
additional demand for $x_1$ is met by using each of the server sets \Ys{j},
\Yp{j} for $25 < j \leq 50$ to cover $x_1$ once.

In the remainder of this section, we establish~\Cref{lem:general-interface}.

\subsubsection{Forming Nets from Filtered Clients}
Roughly speaking, our approach is to extend the proof of~\Cref{lem:spec-full},
i.e. (a) compute a hierarchy of nets $\XL{1} \subseteq \XL{2} \subseteq \dots
\subseteq \XL{k}$, and (b) In each iteration $i = 1 \dots k$, let each client in
\XL{i} add one server to \Yp{i} and one server to \Ys{i}. There is one obstacle
that arises in this approach, and this motivates the following definition. 

\begin{defn}
   \label{defn:threaten}
   We say that client $x_2$ threatens client $x_1$ if
   \begin{itemize}
      \item $\kappa(x_1) > \kappa(x_2)$, and
      \item $\NN{x_1}{\kappa(x_1) - \lfloor \kappa(x_2)/2 \rfloor} \cap \NN{x_2}{\kappa(x_2) - \lfloor \kappa(x_2)/2 \rfloor} \neq \varnothing$.
   \end{itemize}
\end{defn}

Observe that $\NN{x_2}{\kappa(x_2) - \lfloor \kappa(x_2)/2 \rfloor} = \NN{x_2}{l}$, and thus the second condition informally says that some private servers of $x_2$ are also ``inner'' servers of $x_1$. 

To help understand the definition, consider the following example: suppose that $\kappa(x_1) = 100$, $\kappa(x_2) = 50$, and $x_2$ threatens $x_1$. Thus, $\NN{x_1}{75} \cap \NN{x_2}{25} \neq \varnothing$. The plan for our algorithm is that will provide the coverage required by $x_2$ in the first $25$ iterations, and the coverage required by $x_1$ in the first $50$ iterations. Now suppose that $x_2$ is chosen in the net in the first $25$ iterations. This precludes $x_1$ being in the net in these first $25$ iterations. However, we would like to allow $x_1$ to enter the net in iteration $26$, since $x_2$ is essentially finished at this point, whereas $x_1$ is not.

In each of the first $25$ iterations, we would choose two servers for $x_2$, one of which would be a server from  $\NN{x_2}{25} = \NN{x_2}{l}$. We would like at most  one of these two servers to belong to $\NN{x_1}{75}$, so that $x_1$ has enough nearby servers when it later enters the net. However, we cannot ensure this, since the condition $\NN{x_1}{75} \cap \NN{x_2}{25} \neq \varnothing$ means that a private server of $x_2$ can belong to $\NN{x_1}{75}$. 

Therefore, as a preprocessing step, we compute a representative subset $\oX \subseteq X$ in which no client threatens another:

\begin{claim}
\label{cl:filter}
We can compute in polynomial time a subset $\oX \subseteq X$ of clients such that 
\begin{itemize}
\item For any two clients $x_1,x_2$ such that $x_2$ threatens $x_1$, $x_1 \in \oX \implies x_2 \not\in \oX$;
\item For any client $x \in X \setminus \oX$, there is an $x' \in \oX$ such that
$x$ threatens $x'$.
\end{itemize}
\end{claim}

\begin{proof}
Let $\phi$ be any ordering of the clients $X$ such that the $\kappa(\cdot)$ values are non-increasing. Observe that if $x_2$ threatens $x_1$ then $x_2$ occurs after $x_1$ in $\phi$. We initialize $\oX$ to be empty, and assume all clients are initially unmarked. We process each client in $X$ according to the 
ordering $\phi$ as follows: for each cllient $x$, perform the following actions if $x$ is unmarked:
1) add $x$ to $\oX$ 2) mark all clients of $X$ that threaten $x$. 

It is easily checked that the resultant set of clients $\oX$ satisfies the two properties.
\end{proof}

We compute a hierarchy of nets on \oX, instead of $X$. For any client $x \in X
\setminus \oX$, there is a client $x' \in \oX$ such that $x$ threatens $x'$.
Such an $x'$ will help deal with the coverage requirements of $x$. For each
\GL{i}, we define \HL{i} as the subgraph of \GL{i} induced by $\oX$ i.e. $\HL{i}
= \GL{i}[\oX]$. Recall the definition of \GL{i}, and observe that for $1 \leq j
< i \leq k$, if $(x,x')$ is an edge in $\GL{i}$ it is also an edge in $\GL{j}$.
The same holds for edges in \HL{i}.

We will construct a hierarchy of $3$-nets for clients $\oX$, using the family of
graphs \HL{i}, obtaining an anolog of \Cref{cl:ruling}.   

\begin{claim}
\label{cl:ruling-n}
There is a polynomial time algorithm that computes a hierarchy
\[
    \XL{1} \subseteq \XL{2} \subseteq \dots \subseteq \XL{k},
\]
where each $\XL{i} \subseteq \oX$ is a  $3$-net of $\HL{i}$. 
\end{claim}

\subsubsection{Computing Disjoint Server Subsets}

Our algorithm for computing the family \F of server subsets, as stated
in~\Cref{lem:general-interface}, is described in~\Cref{alg:HC-N}. In many ways,
it is analagous to \Cref{alg:HC}, so we only highlight the key differences. One
syntactic feature worth drawing attention to is that index $i$ goes up from $1$
in the for loop in Line 3, as opposed to the for loop in
\Cref{lin:outer-for-start-n} of \Cref{alg:HC} where it decreased starting from
$k$. Thus, iteration $i$ in \Cref{alg:HC-N} corresponds to iteration $k - (i-1)$
in \Cref{alg:HC}. 

In iteration $i$, we consider each client $x_c \in \XL{i}$ in the for loop in
Line 5, but we add the farthest available server in
$\NN{x_c}{\kappa(x_c) - (i - 1)}$ to $\Ys{i}$ only if $\kappa(x_c) \geq 2i$, and
any available server from $\NN{x_c}{l}$ to $\Yp{i}$ only if $\kappa(x_c) \geq 2i -
1$.

\begin{algorithm*}[hbt]
 \caption{\specn}
 \label{alg:HC-N}
 \begin{algorithmic}[1]
 \setcounter{ALC@unique}{0}
    \STATE $l \leftarrow \lceil k / 2 \rceil$ \label{line:l-n}
    \STATE Compute $\XL{1} \subseteq \XL{2} \subseteq \cdots \subseteq \XL{k}$ using Claim~\ref{cl:ruling-n}.
    \FOR {$i = 1$ \TO $l$} \label{lin:outer-for-start-n}
        \STATE Let $\Ys{i} \leftarrow \varnothing, \Yp{i} \leftarrow \varnothing$.
        \FORALL {$x_c \in \XL{i}$} \label{lin:inner-for-start-n}
           \IF {$\kappa(x_c) \geq 2i$}
              \STATE $y_s \leftarrow \mbox{ farthest available server in }
                      \NN{x_c}{\kappa(x_c) - (i - 1)}$. \label{lin:shared-server-n}
              \STATE $\Ys{i} \leftarrow \Ys{i} \cup \{y_s\}$. Mark $y_s$ as not
                 available.
           \ENDIF
           \IF {$\kappa(x_c) \geq 2i - 1$}
              \STATE $y_p \leftarrow \mbox{ any available server in }
                     \NN{x_c}{l}$. \label{lin:priv-server-n}
              \STATE $\Yp{i} \leftarrow \Yp{i} \cup \{y_p\}$. Mark $y_p$ as not
                 available.
           \ENDIF \label{lin:phase3-end-n}
        \ENDFOR 
    \ENDFOR
    \STATE $\F \leftarrow \varnothing$
    \FOR {$i = 1$ \TO $l$} \label{lin:family-begin-n}
       \IF{ $k$ is even or $i < l$} 
           \STATE $\F \leftarrow \F \cup \{\Ys{i}\}$
       \ENDIF
       \STATE $\F \leftarrow \F \cup \{\Yp{i}\}$
    \ENDFOR
 \end{algorithmic}
\end{algorithm*}

Assuming that servers are available when the algorithm looks for them, we can
now establish~\Cref{lem:general-interface}. Fix an $i$ such that $1 \leq i \leq
l$, and assume that $k$ is even. Let $Z = \{x \in X \ | \ \kappa(x) \geq 2i\}$.

To establish part ($3$) of~\Cref{lem:general-interface}, we want to show that
for any client in $Z$, there is a client $\bar{x}$ within $3$ hops of this
client in \GL{i} such that \Ys{i} contains a server from
$\NN{\bar{x}}{\K{i}(\bar{x})}$. Let us first consider the case of a client $x
\in Z$ that also belongs to $\oX$, and hence is a vertex in \HL{i}. Since
$\XL{i}$ is a $3$-net in $\HL{i}$, there is a path $\pi$ in $\HL{i}$ with at
most $2$ edges (and $3$ vertices) that connects $x$ to some vertex $\bar{x} \in
\XL{i}$. Let $\delta(y, \RHI{i}(y))$ be the biggest ball in outer cover
$\RHI{i}$ that serves at least one vertex on path $\pi$. Suppose that it serves
vertex $\hat{x} \in \pi$. ($\hat{x}$ could be the same as $x$ or $\bar{x}$.)
Note that vertices $x'$ in $\HL{i}$ with $i > \lceil \kappa(x')/2 \rceil$ are
isolated. Thus, $\kappa(x') \geq 2i - 1$ for any vertex $x'$ on this path. We
claim that in fact $\kappa(x') \geq 2i$ for any vertex $x'$. Otherwise, since
$\kappa(x) \geq 2i$, there is an edge $(x',x'')$ in $\pi$ such that $\kappa(x')
= 2i - 1$, and $\kappa(x'') \geq 2i$.  Since $(x',x'')$ is an edge in $\HL{i}$,
we have \[ \NN{x'}{\kappa(x') - (i-1)} \cap \NN{x''}{\kappa(x'') - (i-1)} \neq
\varnothing.\] As $i - 1 = \lfloor \kappa(x')/2 \rfloor$, we see that $x'$
threatens $x''$, a contradiction. We conclude that $\kappa(x') \geq 2i$ for any
vertex $x'$ on $\pi$. Thus, $\kappa(\bar{x}) \geq 2i$, and~\Cref{alg:HC-N} adds
a server from $\NN{\bar{x}}{\kappa(\bar{x}) - (i-1)}$ to \Ys{i}
in Line 7.

Now consider an arbitrary client $\shat{x} \in Z \setminus \oX$. There is a
client $x \in \oX$ such that $\shat{x}$ threatens $x$. Thus, $\kappa(x) \geq
\kappa(\shat{x})$,  so $x \in Z \cap \oX$. Furthermore,
   \[ \NN{x}{\kappa(x) - \lfloor \kappa(\shat{x})/2 \rfloor} \cap \NN{\shat{x}}{\kappa(\shat{x}) -\lfloor \kappa(\shat{x})/2 \rfloor} \neq \varnothing.\]
Since $i - 1 \leq \lfloor \kappa(\shat{x})/2 \rfloor$, we have
   \[ \NN{x}{\kappa(x) - (i-1)} \cap \NN{\shat{x}}{\kappa(\shat{x}) - (i-1)} \neq \varnothing.\]
This implies that $(x_1, x)$ is an edge in $\GL{i}$. Using the preceeding
argument, we can prove there is a client $\bar{x} \in X$ that is $2$ hops away
from $x$ in \GL{i}, such that \Ys{i} has a server added to it from
$\NN{\bar{x}}{\K{i}(\bar{x})}$. We can thus infer that $\bar{x}$ is $3$ hops aways from
$x_1$ in \GL{i}. Thus, if $k$ is even, for any client $x$ such that $\kappa(x)
\geq 2i$ there is a client $\bar{x} \in X$ within $3$ hops of $x$ in \GL{i},
such that $\Ys{i} \cap \NN{\bar{x}}{\K{i}(\bar{x})} \neq \varnothing$. 
If $k$ is odd, a similar argument can be made for $i < l$. This completes the
proof of part ($3b$) of~\Cref{lem:general-interface}, predicated on server
availability. 

Part $(3a)$ of~\Cref{lem:general-interface} is established in a similar way. The
argument is actually simpler, because we do not need to argue $\kappa(\bar{x})
\geq 2i$; it suffices that $\kappa(\bar{x}) \geq 2i - 1$. Combined, this
establishes~\Cref{lem:general-interface}, assuming server availability, which we
prove subsequently.

\subsubsection{Server Availability}
\label{sec:sa-n}
In this section, we show that \Cref{alg:HC-N} finds available servers when it looks for them in Line 7 and Line 10. We define
the {\em threshold level} of a client $x \in \oX$ (denoted by $\thr(x)$) as the smallest $i$ for which $x$ belongs to the net $\XL{i}$. (Some clients in $\oX$ may not be part of any of the nets; when we refer to the threshold level of a client, we implicitly assume that it is in some net, in particular, $\XL{k}$.) For client $x \in X$ and iteration $1 \leq i \leq \lceil \kappa(x)/2 \rceil$ of the for loop 
in Line 3, we define $\A{x}{i}$
to be the set of available servers within $\NN{x}{\kappa(x) - (i-1)}$ at the {\em beginning} of iteration $i$. Note that $\A{x}{1} = \NN{x}{\kappa(x)}$.

To establish availability, it suffices to consider clients $x \in \oX$ for which
$\thr(x) \leq \lceil \kappa(x)/2 \rceil$. For a client $x \in \oX$ for which
$\thr(x) > \lceil \kappa(x)/2 \rceil$, the algorithm never looks for available
servers in its neighborhood in Line 7 and Line 10. 

We now show that any such client $x$ has enough available
servers at the beginning of iteration $i = \thr(x)$ of the
outer loop of~\Cref{alg:HC-N}. This argument is where the intricacies 
of the non-uniform MMC and the need for resolving ``threats'' show up. 

\begin{claim}
\label{cl:loop-invariant-n}
Let $x$ be any client in $\oX$ such that $\thr(x) \leq \lceil \kappa(x)/2 \rceil$, and let $i = \thr(x)$. Then 
\begin{clenum}
    \item $|\A{x}{i} \cap \NN{x}{l}| \geq \lceil \kappa(x)/2 \rceil - (i-1)$. 
    \item $|\A{x}{i}| \geq \kappa(x) - 2(i-1)$. 
\end{clenum}
\end{claim}
\begin{proof}
    Consider any iteration $j < i$ of the outer loop in 
    Line 3. The client $x$ itself is not part of the
    net $\XL{j}$. Any client $x' \in \XL{j}$ for which some server
    is chosen in  Line 7 or Line 10
    must satisfy $j \leq \lceil \kappa(x')/2 \rceil$. For such a client 
    $x'$, if $\NN{x}{\kappa(x)-(j-1)} \cap \NN{x'}{\kappa(x')-(j-1)} \neq 
    \varnothing$, then $(x,x')$ is an edge in $\HL{j}$. Since $\XL{j}$ is
    a $3$-net in $\HL{j}$, we conclude that there is at most one
    client $x' \in \XL{j}$ such that (a) some server
    is chosen in  Line 7 or Line 10 for 
    $x'$, and (b) $\NN{x}{\kappa(x)-(j-1)} \cap \NN{x'}{\kappa(x')-(j-1)} \neq \varnothing$. If there is no such client, we can conclude 
    that in iteration $j$, no server in $\NN{x}{\kappa(x)-(j-1)}$ (and 
    thus $\NN{x}{\kappa(x)-(i-1)}$) is made unavailable.

    So let us assume that there is one such client $x'$. Next, we argue that
    $\NN{x}{\kappa(x)-(i-1)} \cap \NN{x'}{l} = \varnothing$. Since server choices
    are made for $x'$ in iteration $j$, we have $j \leq \lceil \kappa(x')/2 \rceil$. 

    First consider the case $i \leq \lceil \kappa(x')/2 \rceil$. Since $x$ and $x'$ are both part of the net $\XL{i}$, $(x,x')$ is not an edge in $\HL{i}$. As $i \leq \lceil \kappa(x')/2 \rceil$ and $i \leq \lceil \kappa(x)/2 \rceil$,
we may conclude that $\NN{x}{\kappa(x)-(i-1)} \cap \NN{x'}{\kappa(x')-(i-1)}
= \varnothing$. Also, since $i \leq \lceil \kappa(x')/2 \rceil$, we have $\NN{x'}{l}
\subseteq \NN{x'}{\kappa(x')-(i-1)}$. Thus, $\NN{x}{\kappa(x)-(i-1)} \cap \NN{x'}{l} = \varnothing$.

    Next, consider the case $i > \lceil \kappa(x')/2 \rceil$. Since $\lceil \kappa(x)/2 \rceil \geq i$, we have that $\kappa(x) > \kappa(x')$. Now, since $x'$ does not theraten $x$, we conclude that $\NN{x}{\kappa(x) - \lfloor \kappa(x')/2 \rfloor} \cap \NN{x'}{\kappa(x') - \lfloor \kappa(x')/2 \rfloor} = \varnothing$. Since $\NN{x}{\kappa(x)-(i-1)} \subseteq \NN{x}{\kappa(x) - \lfloor \kappa(x')/2 \rfloor}$, and $\NN{x'}{\kappa(x') - \lfloor \kappa(x')/2 \rfloor} = \NN{x'}{l}$, we conclude that $\NN{x}{\kappa(x)-(i-1)} \cap \NN{x'}{l} = \varnothing$.

   Thus, in iteration $j$, the server choice made for $x'$ in Line 10 is not from $\NN{x}{\kappa(x)-(i-1)}$, whereas the server choice made for
$x'$ in Line 7 may be from $\NN{x}{\kappa(x)-(i-1)}$. 

Since at most one server from $\NN{x}{\kappa(x)-(i-1)}$ is made unavailable in 
each of the $i- 1$ iterations before iteration $i$, we conclude that $\A{x}{i}
\geq \kappa(x) - (i-1) - (i-1)$. The first assertion of the lemma also follows.

\end{proof}

The next claim says that before every iteration $\thr(x) \leq i \leq \lceil \kappa(x)/2 \rceil$, there are enough available servers in $\NN{x}{\kappa(x)- (i-1)}$. These are iterations in which $x$ itself is part of the net, and the argument is identical to that of \Cref{cl:atmost-two}.

\begin{claim}
\label{cl:atmost-two-n}
Let $x \in \oX$, and let $\thr(x) \leq i < \lceil \kappa(x)/2 \rceil$. 
Then 
\begin{clenum}
    \item $|\A{x}{i + 1}| \geq |\A{x}{i}| - 2$
        \label{cl:atmost-two.1-n}
    \item If $|\A{x}{i + 1}| = |\A{x}{i}| - 2$, then one of the servers 
        in $\A{x}{i} \setminus \A{x}{i+1}$ is the farthest server in $\A{x}{i}$
        from $x$. \label{cl:atmost-two.2-n}
\end{clenum}
\end{claim}

We can now assert our final claim about server availability. The proof 
follows from \Cref{cl:loop-invariant-n} and \Cref{cl:atmost-two-n} using 
arguments very similar to  \Cref{cl:available}.

\begin{claim}
    \label{cl:available-n}
\Cref{alg:HC-N} finds an available server whenever it executes  Line 10 or Line 7.
\end{claim}

This completes our proof of~\Cref{lem:general-interface}.

\subsection{Solving the Non-uniform MMC Problem}
\label{ssec:nummc-algo}
In this section, we describe a constant factor approximation for the non-uniform
MMC problem. Recall that our input consists of two point sets $X$ (clients) and
$Y$ (servers) in an arbitrary metric space $(X \cup Y, d)$, a function $\kappa$
representing the coverage demand of each client, and the constant $\alpha \geq
1$.

\begin{algorithm*}[hbt]
 \caption{\mmcn}
 \label{alg:NUMMC}
 \begin{algorithmic}[1]
 \setcounter{ALC@unique}{0}
    \STATE $k \leftarrow \max_{x \in X}\kappa(x), l \leftarrow \lceil k/2 \rceil$.
    \STATE $\F \leftarrow \specn$.
    \COMMENT{Note that $\F = \{\Ys{1}, \Yp{1}, \Ys{2}, \Yp{2},\dots \}$.}
    \STATE For each $y \in Y$, assign $r(y) \leftarrow 0$.
    \FOR {$i = 1$ \TO $l$} \label{lin:cover-begin-n}
       \IF{ $k$ is even or $i < l$} 
          \STATE Let $r_s$ be obtained by invoking $\text{Cover}(\cdot, \Ys{i}, \alpha)$ for clients $\{ x \in X \ | \ \kappa(x) \geq 2i\}$.
          \STATE Let $r(y) \leftarrow r_s(y)$ for each $y \in \Ys{i}$.
       \ENDIF
       \STATE Let $r_p$ be obtained by invoking $\text{Cover}(\cdot, \Yp{i}, \alpha)$ for clients $\{ x \in X \ | \ \kappa(x) \geq 2i - 1\}$.
       \STATE Let $r(y) \leftarrow r_p(y)$ for each $y \in \Yp{i}$.
    \ENDFOR
    \RETURN The assignment $r: Y \rightarrow \Real^+$ \label{lin:cover-end-n}
 \end{algorithmic}
\end{algorithm*}

Our algorithm first computes a family $\F$ consisting of $k$ pairwise disjoint
subsets of $Y$, using the algorithm of \Cref{lem:general-interface}. It then invokes
$\text{Cover}(\cdot,Y',\alpha)$ using a server subset from \F and a selected
subset of clients as follows. Note that $\F = \{ \Ys{1}, \Yp{1}, \Ys{2},
\Yp{2}, \dots \}$. In the $i$-th iteration of the for loop in
Line 4, we use servers in $\Ys{i}$ to $1$-cover the clients
with coverage demand at least $2 i$, and servers in $\Yp{i}$ to $1$-cover
the clients with coverage demand at least $2i - 1$. Notice that if $k$ is odd
and $i = l$, there are no clients with coverage demand at least $2i$. 

The algorithm then returns $r$, the union of the $k$ covers thus formed, 
which satisfies the coverage demand of each client (as the server
subsets in \F are pairwise disjoint). This union can be thought of as the combined
assignment $r: Y \rightarrow \Real^+$; for a server $y$ not belonging to any
subset in \F, we simply set $r(y)$ to $0$.    

\subsection{Outer Covers}
We generalize the notion of an outer cover as used in the uniform MMC. Let
$\kappa': X \rightarrow \mathbb{Z}^+$ be a coverage function where as usual we
assume $\kappa'(x) \leq |Y|$ for any client $x$. For notational convenience, we
denote $\yd{\kappa'(x)}{x}$, the $\kappa'(x)$-th nearest server of client $x$,
by $\ns{x}{\kappa'}$. 

A $\kappa'$-outer cover is 
an assignment $\rho: Y \rightarrow \Real^+$ of radii to the servers such that 
for each client $x \in X$ for which $\kappa'(x) > 0$, there is a server 
$y \in Y$ such that
\begin{enumerate}
        \item The ball $\delta(y, \rho(y))$ contains $x$ i.e.\ $d(y, x) \leq
            \rho(y)$.
        \item Radius of the ball at $y$ is large, that is, $\rho(y) \geq d(x,
           \ns{x}{\kappa'})$.
\end{enumerate}

Given a level $\kappa'$-outer cover $\rho$, and a client $x \in X$ with $\kappa'(x) > 0$, any server $y$ that 
satisfies the two conditions in the definition above is said to {\em serve} $x$;
we also say that the corresponding ball $\delta(y, \rho(y))$ serves $x$. Observe that we do not require that a client $x$ with $\kappa'(x) = 0$ be covered or served. Also observe that an outer cover of level $i$ is a special case of a
$\kappa'$-outer cover where $\kappa'$ is the constant function that takes on the value $i$.

The following lemma gives a lower bound on the cost of the optimal solution for the non-uniform MMC. It is analogous to \Cref{lem:lbound} and its proof follows by similar arguments. Recall that for $1 \leq i \leq k$, the coverage function $\K{i}: X \rightarrow \mathbb{Z}^+$ is defined by by $\K{i}(x) = \max \{0, \kappa(x) - (i-1)\}$.

\begin{lemma}
    \label{lem:lbound-n}
    Let $r': Y \rightarrow \Real^+$ be any assignment that constitutes a
    feasible solution to the non-uniform MMC. For each 
    $1 \leq i \leq k$, let $\MU{i}$ denote the cost of an optimal $\K{i}$-outer cover. Then
 \[ \sum\limits_{i = 1}^k \MU{i} \leq 3^{\alpha} \cdot
            \cost(r').\]

\end{lemma}

\subsection{Approximation Guarantee}
To obtain an approximation guarantee for \Cref{alg:HC-N}, we first upper bound the cost of the covers returned in iteration $i$ of the for loop in Line 4. 

\begin{claim} 
   Assume that either (a) $k$ is even and $1 \leq i \leq l$, or (b) $k$ is odd
and $1 \leq i < l$. Let $\RHI{i}$ be any  $\K{i}$-outer cover. There is 
a $1$-cover of the clients $\{x \in X \ | \ \kappa(x) \geq 2i\}$ that uses 
servers from $\Ys{i}$ and has cost at most $16^{\alpha} \cdot \cost(\RHI{i})$.
\label{cl:Bi-n}
\end{claim}

\begin{proof}
Let $Z = \{x \in X \ | \ \kappa(x) \geq 2i\}$. 
Consider the set $B$ of balls obtained by expanding each ball in the outer cover
$\RHI{i}$ to $8$ times its original radius. It suffices, as in the proof of
\Cref{cl:Bi}, to show the following claim.

\begin{claim}
For any client $x \in Z$, there is some ball in $B$ that contains $x$ as well as at least one server in $\Ys{i}$.
\label{cl:assert-n}
\end{claim}

%
%
%
%

We now turn to the proof of \Cref{cl:assert-n}. Consider an arbitrary client $x
\in Z$. By~\Cref{lem:general-interface}, there is a path $\pi$ in $\GL{i}$ with
at most three edges that connects $x$ to $\bar{x}$, such that $\Ys{i} \cap
\NN{\bar{x}}{\K{i}(\bar{x}}) \neq \varnothing$. Let $\delta(y,
\RHI{i}(y))$ be the biggest ball in outer cover $\RHI{i}$ that serves at least
one vertex on path $\pi$. Suppose it serves vertex $\hat{x}$. Using the
definition of $\K{i}$, and the way we pick the ball $\delta(y, \RHI{i}(y))$, we
have that for any $x' \in \pi$,
\[
   d(x', \ns{x'}{\K{i}}) \leq \RHI{i}(y).
\]

From the definition of $\GL{i}$, we have that for any edge $(x',x'')$ in $\pi$, 
\begin{equation}
    \label[ineq]{ineq:edge-in-Gi-n}
    d(x',x'') \leq d(x', \ns{x'}{\K{i}}) + d(x'',\ns{x''}{\K{i}}) \leq
    2\RHI{i}(y)
\end{equation} 

By~\Cref{lem:general-interface}, $\NN{\bar{x}}{\K{i}(\bar{x})} \cap
\Ys{i} \neq \varnothing$. Let $\bar{y}$ be an arbitrary server in
$\NN{\bar{x}}{\K{i}(\bar{x})} \cap \Ys{i}$. Clearly, 
\[
   d(\bar{x}, \bar{y}) \leq d(\bar{x}, \ns{\bar{x}}{\K{i}}) \leq \RHI{i}(y).
\]

%

We calculate
\[ d(y,x) \leq d(y,\hat{x}) + \left( \sum\limits_{(x',x'') \in \pi[\hat{x},x]} d(x',x'') \right)  \leq 7 \RHI{i}(y),\]
and
\[ d(y,\bar{y}) \leq d(y,\hat{x}) + \left( \sum\limits_{(x',x'') \in \pi[\hat{x},\bar{x}]} d(x',x'') \right) + d(\bar{x},\bar{y}) \leq 8 \RHI{i}(y).\]

Thus, the ball $\delta(y, 8 \RHI{i}(y))$ contains both $x$ and $\bar{y}
\in \Ys{i}$, completing the proof of \Cref{cl:assert-n}.
\end{proof}

The following claim addresses the cost of the cover obtained using the server set $\Yp{i}$. Its proof is very similar to that of \Cref{cl:Bi-n}. 
\begin{claim} 
   Let $1 \leq i \leq l$ and $\RHI{i}$ be any $\K{i}$-outer cover. 
There is  a $1$-cover of the clients $\{x \in X \ | \ \kappa(x) \geq 2i - 1\}$ 
that uses servers from $\Yp{i}$ and has cost at most $16^{\alpha} \cdot \cost(\RHI{i})$.
\label{cl:Bi-n-v1}
\end{claim}

We can now establish the approxmation guarantee for \Cref{alg:HC-N} and the main result of this section.
 
\begin{theorem}
    \label{thm:non-uniform-result}
    Given point sets $X$ and $Y$ in a metric space $(X \cup Y, d)$ and a 
    coverage function $\kappa$, \Cref{alg:HC-N} runs in polynomial 
    time and returns a $\kappa$-cover of $X$ with cost at most $2 \cdot {(16 \cdot 9)}^{\alpha}$ times that of an optimal $\kappa$-cover.
\end{theorem}

\begin{proof}
It is evident that the algorithm runs in polynomial time. It is also easy to check that the assignment $r$ that it returns is a $\kappa$-cover, that is, each client $x$ is covered at least $\kappa(x)$ times. Let $r'$ be any optimal $\kappa$-cover. By \Cref{lem:lbound-n}, there exists a $\K{i}$-outer cover $\RHI{i}$, for $1 \leq i \leq k$ such that 
\[\sum\limits_{i = 1}^k  \cost(\RHI{i}) \leq 3^{\alpha} \cost(r').\]

From \Cref{cl:Bi-n} and \Cref{cl:Bi-n-v1}, and the fact that $\Cover(\cdot,\cdot,\alpha)$ returns a $3^{\alpha}$ approximation, we conclude that the cost of a $1$-cover that is computed in iteration $i$ of the for loop in Line 4 is at most
$(16 \cdot 3)^{\alpha} \cost(\RHI{i})$. At most two $1$-covers are computed in iteration $i$. Thus,
 
\[
    \cost(r) 
    \leq 2 \cdot (16 \cdot 3)^{\alpha} \cdot \sum\limits_{i = 1}^l \cost(\RHI{i})
    \leq 2 \cdot (16 \cdot 3)^{\alpha} \cdot \sum\limits_{i = 1}^k \cost(\RHI{i})
         \leq 2 \cdot {(16 \cdot 9)}^{\alpha} \cdot \cost(r').
\]
\end{proof}

\end{document}